%% file: main.tex
\documentclass[runningheads, orivec]{llncs}

\usepackage[T1]{fontenc}

\usepackage{microtype}
\usepackage{amssymb}
\usepackage[dvipsnames]{xcolor}
\usepackage{mathtools}
\usepackage{thmtools, thm-restate}

\usepackage{stmaryrd}
\usepackage{MnSymbol}
\usepackage{graphicx}
\usepackage{tasks}
\usepackage{imakeidx}
\usepackage{enumitem}
\usepackage{todonotes}
\usepackage{float}
\ifodd\textwidth
\fi

\usepackage{subfig}

\spnewtheorem{constraint}{Constraint}{}{\itshape}

\usepackage{url}

\makeatletter
\renewcommand\subsubsection{\@startsection{subsubsection}{3}{\z@}%
                       {-5\p@ \@plus -4\p@ \@minus -4\p@}% Formerly -18\p@ \@plus -4\p@ \@minus -4\p@
                       {-0.5em \@plus -0.22em \@minus -0.1em}%
                       {\normalfont\normalsize\bfseries\boldmath}}
\makeatother

\input{macros.tex}

\begin{document}

  \title{Portability of Optimizations from SC to TSO}

  \author{Akshay Gopalakrishnan \and
  Clark Verbrugge}
  %
  % First names are abbreviated in the running head.
  % If there are more than two authors, 'et al.' is used.
  %
  \institute{McGill University, Montreal, Canada \\
  \email{\{akshay.akshay@mail,clump@cs\}.mcgill.ca}}

%
  %  \author{Mark Batty}
  %  \affiliation{%
  %    \institution{University of Kent}
  %    \city{Canterbury}
  %    \country{United Kingdom}}
  %  \email{m.j.batty@kent.ac.uk}
%
    \maketitle

    \begin{abstract}
        It is well recognized that the safety of compiler optimizations is at risk in a concurrent context.
        Existing approaches primarily rely on context-free thread-local guarantees, and prohibit optimizations that introduce a \emph{data-race}. 
        However, compilers utilize global context-specific information, exposing safe optimizations that may violate such guarantees as well as introduce a race.
        Such optimizations need to individually be proven safe for each language model.
        An alternate approach to this would be proving them safe for an intuitive model (like interleaving semantics), and then determine their portability across other concurrent models.  
        In this paper, we address this problem of porting across models of concurrency.
        We first identify a global guarantee on optimizations portable from \emph{Sequential Consistency (SC)} to \emph{Total Store Order (TSO)}.
        Our guarantee is in the form of constraints specifying the syntactic changes an optimization must not incur. 
        We then show these constraints correlate to prohibiting the introduction of \emph{triangular races}, a subset of data-race relevant to TSO.
        We conclude by showing how such race-inducing optimizations relate to porting across \emph{Strong Release Acquire (SRA)}, a known causally consistent memory model.
    \end{abstract}

    %%
    %% Keywords. The author(s) should pick words that accurately describe
    %% the work being presented. Separate the keywords with commas.
    \keywords{Memory Consistency, Compiler Optimizations, Correctness, Sequential Consistency, Total Store Order}

    \input{1.Intro/intro.tex}
    \input{2.Prelim/prelim_intro.tex}
    \input{3.Concrete/conc_intro.tex}

    \input{4.Disc/disc_intro.tex}

    \bibliography{ref}

    \input{5.Appendix/appendix_intro.tex}

  \end{document}

%% file: macros.tex
\newcommand{\po}{\textcolor{BlueViolet}{po}}
\newcommand{\rf}{\textcolor{Green}{rf}}

\newcommand{\mo}{\textcolor{Red}{mo}} 
\newcommand{\hb}{\textcolor{NavyBlue}{hb}}
\newcommand{\rb}{\textcolor{RubineRed}{rb}}
\newcommand{\rfe}{\textcolor{Green}{rfe}} 
\newcommand{\rfi}{\textcolor{Green}{rfi}}

\newcommand{\we}{\textcolor{BurntOrange}{we}}
\newcommand{\wi}{\textcolor{BurntOrange}{wi}}
\newcommand{\tuwri}{\textcolor{BurntOrange}{tuwri}}

\newcommand{\psf}[3]{\textit{psafe}(#1, #2, #3)}

\newcommand{\rel}{\textcolor{CadetBlue}{rel}}

\newcommand{\comp}[2]{\text{comp}\langle #1, #2 \rangle}

\newcommand{\weak}[2]{\text{weak}\langle #1, #2 \rangle}

\newcommand{\pwc}[1]{\textit{pwc}(#1)}

\newcommand{\wf}[1]{\text{wf}(#1)}

\newcommand{\inarrc}[1]{\begin{array}{@{}c@{}}#1\end{array}}
\newcommand{\inarr}[1]{\begin{array}{@{}l@{}}#1\end{array}}
\newcommand{\inarrII}[2]{
  \begin{array}{@{}l@{~~}||@{~~}l@{}}
    \inarr{#1}&\inarr{#2}
  \end{array}
}
\newcommand{\inarrIII}[3]{
  \begin{array}{@{}l@{~~}||@{~~}l@{~~}||@{~~}l@{}}
    \inarr{#1}&\inarr{#2}&\inarr{#3}
  \end{array}
}

%% file: 1.Intro/intro.tex
\section{Introduction}

    Compilers today are primarily responsible for the performance of any program.
    This is mainly due to plethora of optimizations they perform, which help reduce computation and memory costs significantly.
    However, performing them on concurrent programs has greatly been restricted due to safety concerns. 
    The core reason for this is the underlying concurrent semantics (memory model), which break the guarantees of functional behavior of programs assumed by compilers~ \cite{AdveS}.
    
    Existing efforts have no doubt addressed this up to some extent~\cite{MoiseenkoP}. 
    Optimizations which can be performed are subject to the underlying thread-local guarantees provided by the model.
    For instance, any code motion involving shared memory accesses are prohibited under Sequential Consistency (\emph{SC}). 
    Whereas any motion involving only write-read reordering is permitted under Total Store Order (\emph{TSO}).
    Optimizations are also prohibited from introducing data-races, and programs are often required to be \emph{data-race-free}.
    
    However, writing race-free programs is hard, and locating data-race errors in a program is often difficult, let alone debug.
    To add, the race-free restriction may be too strict for compilers, and weaker guarantees are shown to suffice~\cite{OwensTrftso,Marino}. 
    Moreover, thread-local restrictions can be overly conservative, and may prove detrimental when designing optimizations leveraging aspects of concurrency.
    
    \input{1.Intro/motivation.tex}

    Examples like above represent the kind of global optimizations that can be permitted for \textit{TSO}.
    However, as we see, doing them depends on the right program context.
    Identifying a precise constraint over \emph{all possible} program contexts specific to \textit{TSO} may be non-trivial.
    
    We adopt an alternate approach that may be more suitable to put into practice.
    We instead determine constraints on optimizations which are portable from \textit{SC} to \textit{TSO}. 
    Working towards identifying such a constraint has led to the following contributions 
    \begin{enumerate}
        \item We find a suitable thread-global syntactic constraint on optimizations that do not eliminate/introduce writes.
        \item We show that our constraint correlates to disallowing optimizations resulting in a new \emph{triangular-race}, a subset of \emph{data=races}.
        \item We discover that porting \textit{SC} optimizations that introduce such \emph{triangular races} in addition translates to porting across a Causally Consistent memory model known as Strong Release Acquire (SRA)~\cite{nickSRA}.
    \end{enumerate}

%% file: 1.Intro/motivation.tex
    As an example, consider the program $P$ as described in Fig~\ref{intro:ex1}, where $y$ is shared memory initialized to $0$, and $c$ is thread-local. 
    $\ldots$ represent any thread-local computation in each thread.
    First thread writes to $y$, whereas the second spins until it observes the updated value.

    \input{1.Intro/mot_ex1.tex}

    From the compiler's perspective, the following observation can be made, subject to a fair scheduler 
    \begin{itemize}
        \item The second thread will eventually see the updated value of $y$ (memory fairness \cite{LahavPod}).
        \item The second thread will only see the updated value of $y$ as 1 (value range analysis over $y$).
        \item Thus, $while(!c)$ will eventually fail to hold (thread fairness).
    \end{itemize}
    With this, the compiler can assert spinning on the loop may not be required if the updates value of $y$ is guaranteed to be in memory before it is read.
    To do this, it can simply sequence (inline) the write $y=1$ to take place before both the threads\footnotemark.
    This is shown as $P'$ in Fig~\ref{intro:ex1}.
    Further, it can perform simple store-to-load forwarding, which results in the loop condition to be false from the start.
    This gives us the optimized program $P''$ in Fig~\ref{intro:ex1}.
    Such an optimization is safe under both \textit{SC} and \textit{TSO} semantics, the original and the optimized program both terminate with the same state of memory $c=1$.  
    \footnotetext{
        As a C code, this would translate to creating a thread which does $y=1$ followed by creating two threads which do the other computations.
    }
    
    However, the situation is slightly different if the optimization is done in the presence of another concurrent context.
    In Fig~\ref{intro:ex2} left, the original program cannot terminate with $a=1 \wedge b=0 \wedge c=1 \wedge d=0$, both under \textit{TSO} and \textit{SC} semantics.
    The thread-local state $c=1 \wedge d=0$ implies $y=1$ is visible to all threads before $x=1$.
    Hence, the reads to $x$ and $y$ done after loop termination cannot observe the stale value of $y$ (as $b=0$ shows).
    
    \input{1.Intro/mot_ex2.tex}
    
    The same above reasoning is true for both the optimized programs (middle and right) under \textit{SC} semantics, thereby disallowing the outcome in question.
    However, the situation is different under \textit{TSO}: the program in Fig~\ref{intro:ex2} middle can exhibit the outcome \cite{OwensS}.
    Reading $c=1$ from $y=1$ does not enforce a global ordering between the writes $x=1$ and $y=1$\footnotemark.
    The same reasoning applies for the final optimized program (Fig~\ref{intro:ex2} right). 
    Thus, the same optimization is unsafe when considered under a different program context.
    \footnotetext{
        Operationally, the read value can be fetched from the FIFO write buffer where $y=1$ is first committed. 
        The buffer can be flushed at any latter time, and only then would imply a global memory ordering with $x=1$. 
    }

    %We can go a bit further, and design optimizations that can identify redundant atomic accesses, like in Fig~\ref{blah} (left).
    %Here, $x, y, z$ represent shared memory, the rest are thread-local.
    %
    %\input{1.Intro/mot_ex3.tex}
%
    %From the compiler's perspective, the following observations can be made 
    %\begin{itemize}
    %    \item The $cas$ instruction will always succeed (value range analysis over $z$).
    %    \item The final state of memory will always have $z=2$.
    %\end{itemize}
    %With this, the compiler can assert the redundant dependency between the $cas$ instruction and the write $z=2$, thereby removing it completely, as in Fig~\ref{blah} (right).
    %
    %Such an optimization is safe under \textit{SC} semantics: intuitively, the $cas$ places no additional ordering constraint over the writes visible to other threads.  
    %However, the same cannot be said for \textit{TSO}, with the concurrent context.
    %The optimized program (right) represents a store buffering shape \cite{OwensS}, thereby permitting the outcome $a=0 \wedge b=0$.

%% file: 1.Intro/mot_ex1.tex
%An example where inlining can be inferred by the compiler using value range analyses

\begin{figure*}[htbp]
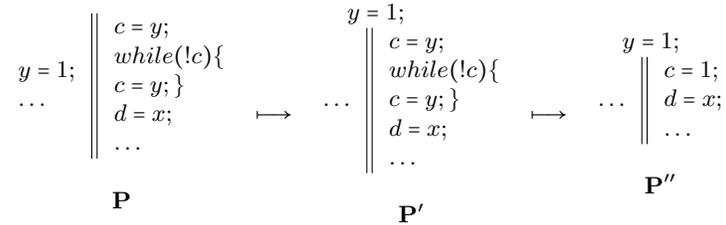

        
    \begin{equation*}
        \inarrc{
        \inarrII{
            y=1;\\
            \ldots
        }{
            c=y;\\
            while(!c)  \{ \\
                c=y;
            \} \\ 
            d=x; \\ 
            \ldots  
        }\\
        \\
        \bf{P}
        }
        \quad\longmapsto\quad
        \inarrc{
        \inarr{
            \quad y=1; \\
            \inarrII{
                \ldots
            }{
                c=y;\\
                while(!c) \{ \\
                    c=y;
                \} \\ 
                d=x; \\
                \ldots              
            }
        }
        \\
        \\
        \bf{P'}}
        \quad\longmapsto\quad
        \inarrc{
        \inarr{
            \quad y=1; \\
            \inarrII{
                \ldots
            }{
                c=1;\\
                d=x; \\
                \ldots              
            }
        }\\
        \\
        \bf{P''}}
    \end{equation*}
    \caption{Optimizations leveraging thread/memory fairness and value range analysis.}
    \label{intro:ex1}
\end{figure*}

%% file: 1.Intro/mot_ex2.tex
%Example where the same inlining optimization is unsafe given an additional concurrent context

\begin{figure*}[htbp]
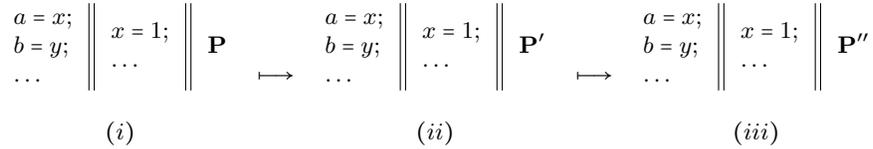

        
    \begin{equation*}
        \inarrc{
        \inarrIII{
            a=x; \\ 
            b=y; \\
            \ldots
        }{
            x=1;\\
            \ldots
        }{
            \bf{P}    
        } \\
        \\
        (i)
        }
        \quad\longmapsto\quad
        \inarrc{
        \inarrIII{
            a=x; \\ 
            b=y; \\
            \ldots
        }{
            x=1;\\
            \ldots
        }{
            \bf{P'}    
        }\\
        \\
        (ii)
        }
        \quad\longmapsto\quad
        \inarrc{
        \inarrIII{
            a=x; \\ 
            b=y; \\
            \ldots
        }{
            x=1;\\
            \ldots
        }{
            \bf{P''}    
        }\\
        \\
        (iii)
        }
        \end{equation*}
    \caption{Under different program context, optimization to $P''$ is unsafe under $TSO$.}
    \label{intro:ex2}
\end{figure*}

%% file: 2.Prelim/prelim_intro.tex
\section{Preliminaries}

    We give a brief overview of the formal elements involved in this work. 
    Sec~\ref{subsec:pre-trace} introduces the pre-trace model, and Sec~\ref{subsec:effect} the representation of optimizations, both from our previous work \cite{gopaltransf}. 
    Sec~\ref{subsec:axiomatic} introduces the axiomatic models for \textit{SC} and \textit{TSO} we use. 
    
    The language is given in Fig~\ref{prelim:lang}.
    A program $\textit{prog}$ is a parallel composition of individual sequential programs $\textit{sp}$, each of which are associated with a thread id $t$ and a sequence of actions $p$.
    An action is either a memory event $st$ or a conditional branch code block (for simplicity, loops are not included).
    Memory events can be a read from ($a=x$), write to ($x=e$), or read-modify-write ($rmw(x, a, v)$) to some shared memory $x$.
    Write events are also associated with a value $v$ or some thread-local variable $a$ (for simplicity we keep it as integers). 
    Let $tid$ be a mapping from memory events to their associated thread-id.
    \begin{figure*}
        \subfloat{
            $\begin{aligned}
                &prog := sp || prog \ | \ sp \\
                &sp := t:p \\ 
                &p := st \ | \ p;p; \ | \ \text{if}(cond) \ \text{then} \ \{p\} \ \text{else} \ \{p\} \ \\ 
                &st := a\!=\!x; \ | \ x\!=\!v; \ | \ rmw(x, a, v);
            \end{aligned}$
        }        
        %    \quad\quad
        \vrule
        \subfloat{
            $\begin{aligned}
                &e := a \ | \ v \\
                &cond := \text{true} \ | \ \text{false} \ | \ a == v \ | \ a !\!= v \\   
                &\text{domains} := v \in \mathbb{Z} \ | \ t \in (\mathbb{N}\!\cup\!\{0\})        
            \end{aligned}$
        }
        \caption{Programs - Adapted from \cite{gopaltransf}}
        \label{prelim:lang}
    \end{figure*}
    
    \input{2.Prelim/pre-trace.tex}

    \input{2.Prelim/sc-tso.tex}

%% file: 2.Prelim/pre-trace.tex
\subsection{Pre-Trace Model}

    \label{subsec:pre-trace}
    The pre-trace model at its core, reflects the compiler's perspective of a program (see \cite{gopaltransf} for full formal details).
    Programs are viewed in their abstract forms, with syntactic order ($\po$) and conditional dependencies (Fig~\ref{base:Pr-to-P} middle).
    As opposed to traditional mapping of abstract programs to concurrent executions, they are first mapped to a set of pre-traces instead (Fig~\ref{base:Pr-to-P} left and right). 
    
    A pre-trace is a sub-program without conditionals that contains a possible execution path of the actual program. 
    In Fig~\ref{base:Pr-to-P}, Pre-trace $P1$ (left) is derived by taking the left conditional branch path, whereas $P2$ (right) is derived by choosing the right. 
    Notice that both $P1$, $P2$ do not have the read values restricted, and thus, represent an over-approximation of the actual program behavior\footnotemark. 
    \begin{figure*}[htbp]
        \centering
        \includegraphics[scale=0.6]{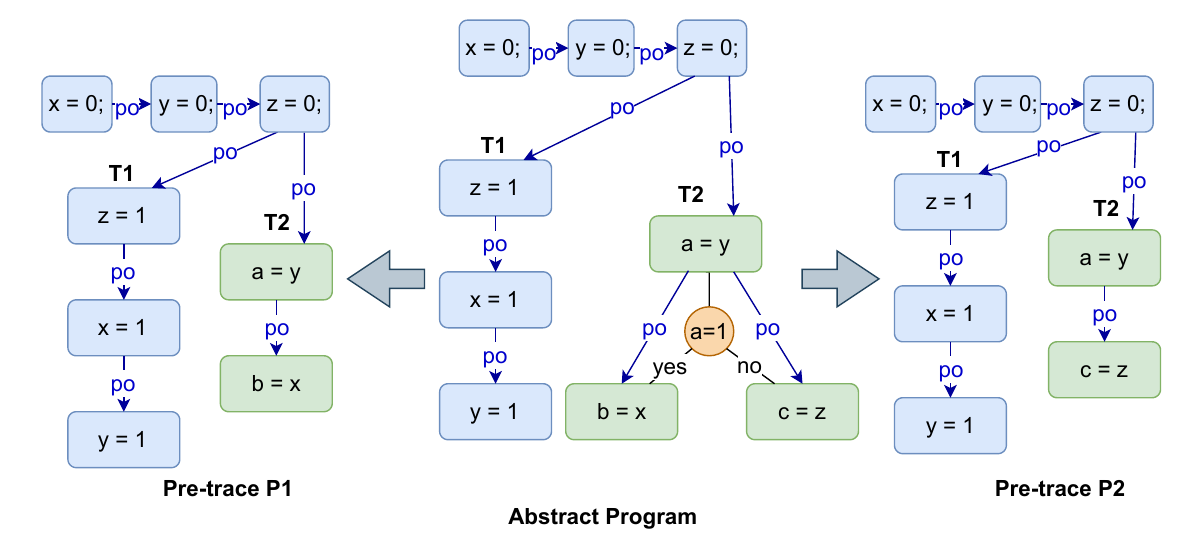}
        \caption{Abstract program (middle) mapped to its possible pre-traces P1 (left) and P2 (right) (Adapted from Fig 11 of \cite{gopaltransf}).}
        \label{base:Pr-to-P}
    \end{figure*}
    \footnotetext{
        Synthesizing pre-traces for a given language from abstract programs would also require the language's appropriate sequential semantics. 
        The results in this paper are independent of the choice of such semantics.
    }
    Let $\po(P)$ represent the syntactic order in $P$. 
    Let $st(P), w(P), r(P), u(P)$ represent the set of events/write/read/read-modify-write events respectively of pre-trace $P$ ($u(P) = r(P) \cap w(P)$).
    The above can be filtered via an optional subscript $loc$, giving the set of events operating on the same shared memory $loc$.
        
    To represent executions $E$, each pre-trace $P$ is associated with a reads-from set $\rf$ ($w_{loc}(P) \times r_{loc}(P)$), denoting the source (write event) of the read value, and memory order set $\mo$ ($w(P) \times w(P)$), denoting the order in which writes are propagated to main memory.  
    Let $p(E)$ give the pre-trace of $E$ and $\rf(E), \mo(E)$ give the $\rf$, $\mo$ relations of $E$ respectively.
    A candidate execution $E$ (like Fig~\ref{base:P-to-E} $E1$, $E2$) of a pre-trace is one where 
    \begin{itemize}
        \item Each read has some value - $\forall r \in r(p(E)) \ . \ \exists w \in w(p(E)) \ . \ (w,r) \in \rf(E)$.
        \item Each read has exactly one source write - $(w,r) \in \rf(E) \wedge (w', r) \in \rf(E) \implies w=w'$.
        \item Propagation order is total - $\mo(E) \ \text{total order}$. 
        \item Propagation order per-shared location is strict $\forall loc \ . \ \mo_{loc}(E) \ \text{strict}$.
    \end{itemize} 
    \begin{figure*}[htbp]
        \centering
        \includegraphics[scale=0.6]{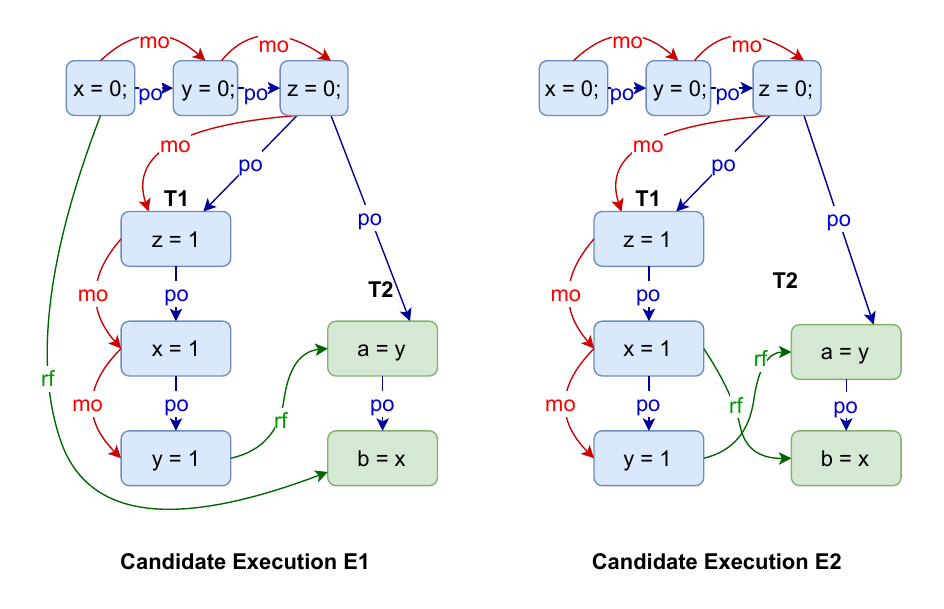}
        \caption{Pre-trace $P1$ from Fig~\ref{base:Pr-to-P} annotated with $\rf$ and $\mo$ edges to represent a two possible candidate execution $E1$ and $E2$ (Adapted from Fig 12 of \cite{gopaltransf}).}
        \label{base:P-to-E}
    \end{figure*}
    An \emph{observable-behavior} of a candidate execution is the final state of memory ($\rf$ relations and $\mo$ maximal per-location).
    $\langle P \rangle$ represents the set of all candidate executions of pre-trace $P$. 
    
    %------------------------------------------------------------------------------------------------

        The concurrency semantics or memory model are a set of constraints (consistency rules) on these candidate executions.
        They are typically in the form of acyclic/irreflexivity constraints on the relations between memory events in the execution.
        Executions that adhere to the rules are deemed \emph{consistent}, whereas those which do not are \emph{inconsistent}. 
        For example, consider the candidate executions $E1$, $E2$ from Fig~\ref{base:P-to-E}.
        $E1$ indicates a clear `message passing' violation: on seeing $a=1$, $b=0$ should not be possible.
        A memory model prohibiting such message passing violations can simply be `$(\rf^{-1} \cup \mo_{loc} \cup \po \cup \rf)$ \emph{acyclic}'.
        With this constraint, $E1$ is inconsistent, however, $E2$ is consistent.  
        %We denote $c_{M}(E)$ to say $E$ is consistent under memory model $M$.
       
        $\llbracket P \rrbracket_{M}$ represents the set of candidate executions of pre-trace $P$ consistent under memory model $M$.
        $I\langle P \rangle_{M}$ represents the inconsistent set. 
        A memory model $M$ is weaker than $B$ ($\weak{M}{B}$) if the set of consistent executions under $B$ is lesser than those under $M$ ($\forall P \ . \ \llbracket P \rrbracket_{B} \subseteq \llbracket P \rrbracket_{M}$).

%----------------------------------------------------------------------------------------------------
\subsection{Optimization as transformation-effects}

    \label{subsec:effect}

        Pre-traces are useful in decomposing an optimization into several \textit{transformation-effects}, each on a different pre-trace of the source program.
        These effects are defined as syntactic changes, involving a set of program orders removed/added ($\po^{-}, \po^{+}$) and set of memory events removed/added ($st^{-}, st^{+}$).
        A transformation effect $tr$ modifying $P$ to $P'$ is denoted as $P \mapsto_{tr} P'$. 
        Recalling the example from Fig~\ref{intro:ex1}, the optimizations have at least the following effects for each pre-trace \footnotemark.
        \begin{enumerate}
            \item $P \mapsto_{tr1} P'$ 
                \begin{tasks}(3)
                    \task $st^{-} = st^{+} = \po^{-} = \phi$. 
                    \task* $\po^{+} = \{ (y=1, c=y), (y=1, a=x) \}$.
                \end{tasks}           
            \item $P' \mapsto_{tr2} P''$ 
                \begin{tasks}(3)
                    \task $st^{-} = \{ c=y \}$.
                    \task $st^{+} = \phi$.
                    \task $\po^{+} = \phi$.
                    \task* $\po^{-} = \{ (y=1, c=y), (c=y, a=x) \}$.
                \end{tasks} 
        \end{enumerate} 

        \footnotetext{
            At least is used to denote that a loop results in multiple pre-traces, one for every number of iteration.
        }
       
        %\begin{figure*}[htbp]
        %    \centering
        %    \includegraphics[scale=0.6]{2.Prelim/TransfEffectEx.pdf}
        %    \caption{Missing Caption}
        %    \label{base:reord-eg}
        %\end{figure*}
        
        An optimization is considered \textit{safe} for a program under a memory model $M$ if the set of consistent observable behaviors do not increase.
        At the level of pre-traces, an effect $P \mapsto_{tr} P'$ is safe if the \emph{set of behaviors} specified by $\llbracket P \rrbracket_{M}$ does not increase.
        Two executions $E1, E2$ have the \emph{same observable behavior} ($E1 \sim E2$) if the set of shared memory read/write(s) common to both have same $\rf$/$\mo$ relations respectively \footnotemark.  
        This is used to compare set of behaviors: $A \sqsubseteq B$ denotes a set of executions $A$ are contained in $B$.
        Thus, a transformation-effect $P1 \mapsto_{tr} P2$ is safe under memory model $M$ ($\psf{M}{tr}{P}$), if $\llbracket P' \rrbracket_{M} \sqsubseteq \llbracket P \rrbracket_{M}$.

        \footnotetext{
            Note that we relate the entire $\mo$ relation, not just the maximal elements.
        }
        
    %---------------------------------------------------------------------------------------------------------------------

%% file: 2.Prelim/sc-tso.tex
\subsection{Axiomatic Version of SC and TSO}

    \label{subsec:axiomatic}     
    We adopt the declarative (axiomatic) style of both \textit{SC} and \textit{TSO} described using irreflexivity constraints \cite{LahavV}.

    \paragraph*{Relational notations}
    Given a binary relation $R$, let $R^{-1}$, $R^{?}$, $R^{+}$ and $[E]$ represent inverse, reflexive, transitive closure and identity relation over a set $E$ respectively. 
    Let $R1;R2$ represent sequential composition of two binary relations.
    A relation $R1;R2$ \emph{forms a cycle} (or simply $R1;R2$ \emph{cycle}) if there exists a cyclic path $[a];R1;[b];R2;[a]$ which is non-empty (reflexive).
    %We assume all sequential compositions are irreflexive, unless otherwise stated or proven to be so.
    Lastly, we say `we have $[a];R1;[b]$' if the relation is non-empty.  

    \paragraph*{Additional relations}
    Let \textit{read-from-internal} ($\rfi$) represent the subset of $\rf$ such that both the write and read are of the same thread. 
    Let \textit{reads-from-external} ($\rfe$) be the rest.
    Let \textit{happens-before} ($\hb$) be $(\po \cup \rf)^{+}$.
    Let \textit{memory-order-loc} ($\mo_{|loc}$) represent the memory order between writes to same memory, and $\mo_{!loc}$ the rest.
    Further, let \textit{memory-order-ext} ($\mo_{ext}$) represent the memory order between writes not ordered by any $\hb$ relation ($\mo \setminus (\mo \cap \hb))$.
    %Further, let \textit{memory-order-internal} ($\mo_{int}$) represent the memory order between writes having an $\hb$ relation, and \textit{memory-order-external} ($\mo_{ext}$) the remaining set.
    Finally, let \textit{reads-before} ($\rb$) represent the sequential composition $\rf^{-1};\mo_{|loc}$.

    Using the above elements \textit{SC} has the following set of constraints \cite{LahavV}.   
    \begin{definition}
        \label{def:sc-model}
        An execution $E$ is consistent under \textit{SC} if the following rules hold 
        \begin{tasks}(3)
            \task $\mo$ strict total order.
            \task $\hb$ irreflexive.
            \task $\mo;\hb$ irreflexive. %not $\mo_{ext}$ as we want to even disallow $\mo_{int};\hb$ being reflexive.
            \task $\rb;\hb$ irreflexive.
            \task $\rb;\mo$ irreflexive.
            \task $\rb;\mo;\hb$ irreflexive.
        \end{tasks}
    \end{definition}
    and the fragment of \textit{TSO} not involving fences have the following set of constraints. 
    \begin{definition}
        \label{def:tso-model}
        An execution $E$ is consistent under \textit{TSO} if the following rules hold 
        \begin{tasks}(2)
            \task $\mo$ strict total order.
            \task $\hb$ irreflexive.
            \task $\mo;\hb$ irreflexive.
            \task $\rb;\hb$ irreflexive.
            \task $\rb;\mo$ irreflexive.
            \task $\rb;\mo;\rfe;\po$ irreflexive.
            \task $\rb;\mo;[u];\po$ irreflexive.
        \end{tasks}    
    \end{definition}
    We can, for \textit{TSO}, consider fences equivalent to a read-modify-write to a shared location not used \cite{LahavV,nickSRA}. 
    Note that \textit{TSO} and \textit{SC} have Rules (a) through (e) the same, with Rule (f) of \textit{SC} being divided into two rules (f) and (g).

%% file: 3.Concrete/conc_intro.tex
\section{Optimizations: From SC to TSO}

    We now address the porting problem using transformation effects over pre-traces.
    Sec~\ref{subsec:complete} states our desired property to be proven over effects.
    Sec~\ref{subsec:constraints} goes over identifying the constraints over optimizations that enable proving said property.
    Sec~\ref{subsec:main-result} states our main result, with examples giving intuition behind the proof.   
    Sec~\ref{subsec:race-sra} relates our result to triangular races and causal consistency.   

    \input{3.Concrete/complete.tex}

    \input{3.Concrete/constraints.tex}

    \input{3.Concrete/tso.tex}

    \input{3.Concrete/races.tex}

    \input{3.Concrete/sra.tex}

%% file: 3.Concrete/complete.tex
\subsection{From Optimizations to Effects}

    \label{subsec:complete}
    Our objective of identifying a set of optimizations portable to TSO can instead be viewed at the level of effects. 
    An optimization of a program is portable from memory model $B$ to $M$ if all its constituent safe effects under $B$ are also safe in $M$.
    Lifting this notion to any program, we converge to the following property between models defined in \cite{gopaltransf}\footnotemark.
    \begin{definition}
        \label{def:complete}
        Memory model $M$ is \emph{complete} w.r.t. $B$ ($\comp{M}{B}$) if 
        \begin{align*}
            \forall P \mapsto_{tr} P' \ . \ \psf{B}{tr}{P} \implies \psf{M}{tr}{P}. 
        \end{align*}
    \end{definition} 

    \footnotetext{
        The origins behind Def~\ref{def:complete} stem from preserving safety of optimizations across memory models that are incrementally built by adding desired optimizations.
    }

    %Proving complete general theorem
    Proving $\comp{M}{B}$ between models is infeasible by enumerating the set of effects and pre-traces, as both are unbounded. 
    %However, the contrapositive can be proven, via quantifying the effects using the constraints/consistency rules of the models involved (Theorem 4.4 by \cite{gopaltransf}).
    Proving the contrapositive however, can be done using the constraints specific to $B$ and $M$. 
    At its core, it requires showing $\psf{B}{tr}{P}$ to be false for any $tr$ such that $\neg \psf{M}{tr}{P}$.
    The following lemma quantifies all such $tr$\footnotemark.
    %This will allow us to identify precise $tr$ in the form of $\po^{+/-}$ and $st^{+/-}$.
    %The set is simply those $tr$ which are unsafe under $M$ but safe under $B$, which is described by the following lemma.
    \begin{restatable}{lemma}{lemcompgen}
        \label{lem:comp-gen-cond}
        Given $\weak{M}{B}$, a transformation-effect $P \mapsto_{tr} P'$ and $\langle P' \rangle \sqsubseteq \langle P \rangle$, $tr$ is unsafe under $M$, but safe under $B$ if for every $E'$ such that, 
        \begin{align*}
            E' \in \llbracket P' \rrbracket_{M} \wedge \nexists E \in \llbracket P \rrbracket_{M} \ . \ E \sim E'.
        \end{align*}    
        we have 
        \begin{tasks}(3)
            \task*(2) $\forall E \in \langle P \rangle \ . \ E \sim E' \implies E \in I\langle P \rangle_{M}.$ 
            \task $E' \in I \langle P' \rangle_{B}$.
        \end{tasks}    
    \end{restatable}

    \footnotetext{
        While Theorem 4.4\cite{gopaltransf} addresses the same issue, its core proof relies on our formulated lemma. 
    }

    \subsubsection{Requirements for $M$, $B$, $tr$}
    
        Lemma~\ref{lem:comp-gen-cond} relies on two conditions.
        First, $M$ must be weaker than $B$.
        This is true when $M=TSO$ and $B=SC$ \cite{OwensS}.  
        Second, we require $\langle P' \rangle \sqsubseteq \langle P \rangle$. 
        Intuitively, it means that any behavior $P'$, without any memory model constraints must also be observable in $P$.
        This disallows write-introduction effects: $\wi$ ($st^{+} \cap w(P') = \phi$), forming our first constraint over $tr$.
        We believe this to be a reasonable constraint, given that any write introduced may be read, leading to an unobservable outcome in the original program. 
        Lastly, we consider $tr$ to have no write-elimination effects: $\we$ ($st^{-} \cap w(P) = \phi$), deferring it to future work.
        %We refer disallowing both write introduction and elimination effects as disallowing \emph{write modification}.

%% file: 3.Concrete/constraints.tex
\subsection{Identifying Constraint over Effects}

    \label{subsec:constraints}
    The example from Fig~\ref{intro:ex2} shows global SC optimizations are portable to TSO only under appropriate program context.
    Instead of identifying such contexts, we can identify constraints on the set of effects that make it unsafe to port. 
    These constraints will be a mixture of the original program context and the effects. 

    We start by viewing the same optimization at the level of pre-traces in Fig~\ref{conc:c1}. 
    \begin{figure*}[h]
        \centering
        \includegraphics[scale=0.6]{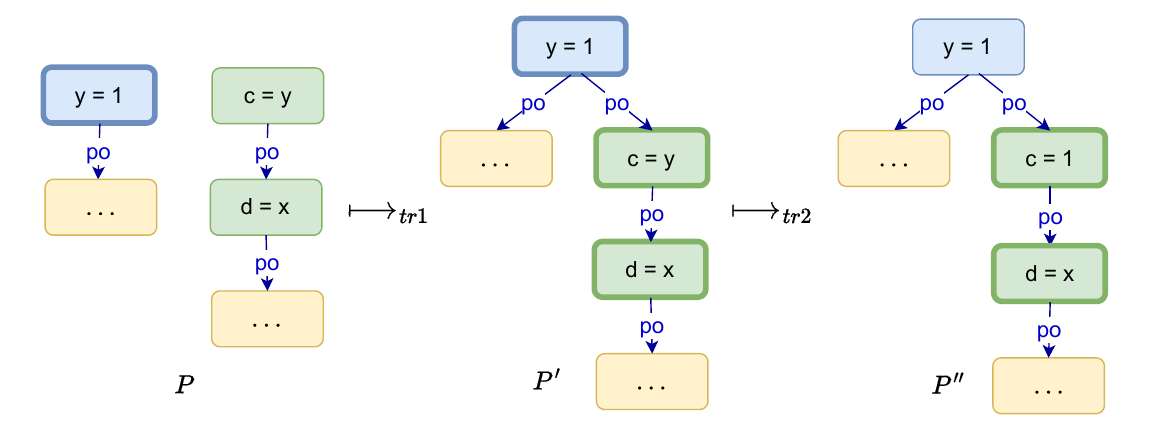}
        \caption{Example optimization of Fig~\ref{intro:ex1} revisited as effects over pre-traces.}      
        \label{conc:c1}
    \end{figure*}
    The optimizations $tr1$ and $tr2$ have the following constraints over their corresponding effects.
    \begin{enumerate}
        \item $P\mapsto_{tr1} P'$ - $(y=1, c=y), (y=1, d=x) \in \po^{+}$.
        \item $P' \mapsto_{tr2} P''$ - $(c=y \in st^{-}) \wedge ((y=1, c=y), (c=y, d=x) \in \po^{-})$.  
    \end{enumerate}
    
    We observe that the effect $tr2$ a form of read-after-write elimination, which is safe under \textit{TSO} irrespective of program context \cite{OwensS}.
    Thus, the constraint is relevant over $tr1$.
    \begin{constraint}
        \label{cons:c0}
        The effect $P \mapsto P'$ may be unsafe under \textit{TSO} if the following hold 
        \begin{tasks}(2) 
            \task $(w_{y}, r_{y}), (w_{y}, r_{x}) \in \po^{+}$.
            \task $(r_{y}, r_{x}) \notin \po^{-}$.
        \end{tasks}
    \end{constraint}
    We can further refine this constraint, noting that without the concurrent context in Fig~\ref{intro:ex2}, $tr1$ is safe under \textit{TSO}. 
    Specifically, we note the context $x=1$, and two candidate executions $E \sim E'$ of the pre-traces as in Fig~\ref{conc:c2}.
    \begin{figure*}[h]
        \centering
        \begin{minipage}{\columnwidth}
            \subfloat[a][
                Pre-trace effect $P \mapsto_{tr1} P'$. 
            ]{\label{conc:c21}\includegraphics[scale=0.6]{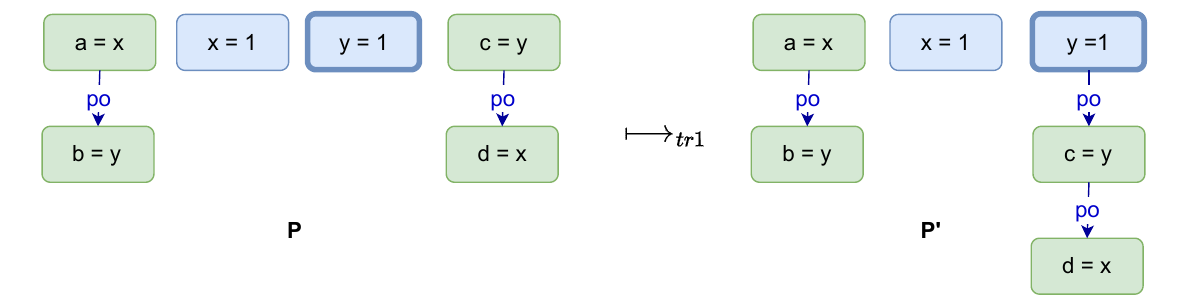}}    
        \end{minipage}
        \begin{minipage}{\columnwidth}
            \subfloat[b][
                Candidate executions of $P$ and $P'$ for the outcome $a=1 \wedge b=0 \wedge c=1 \wedge d=0$. 
            ]{\label{conc:c22}\includegraphics[scale=0.6]{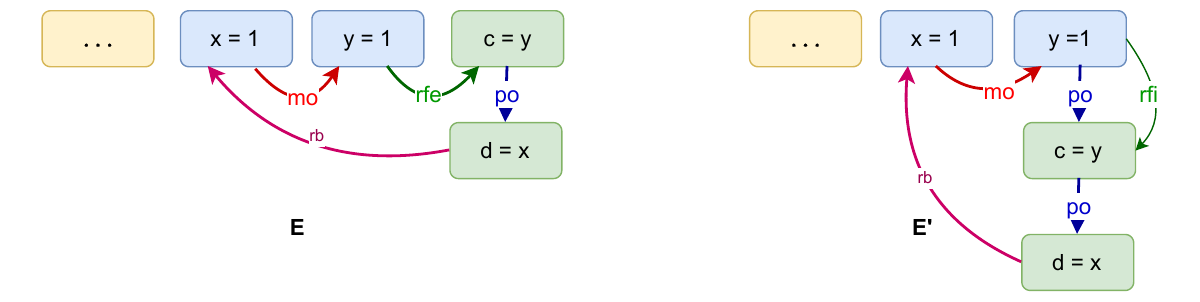}}                 
        \end{minipage}\
        \caption{Example optimization of Fig~\ref{intro:ex2} revisited as effects over pre-traces.}
        \label{conc:c2}
    \end{figure*}
    Notice that $E$ is inconsistent under \textit{TSO} as $[d=x];\rb;\mo;\rfe;\po$ forms a cycle.
    On the other hand $E'$ is consistent as $[d=x];\rb;\mo;\po$ cycles are permitted under \textit{TSO}.
    This implies the subset of $P \mapsto_{tr1} P'$, having a write $w_{x} \in st(P)$ is unsafe under \textit{TSO}. 
    \begin{constraint}
        \label{cons:c1}
        The effect $P \mapsto P'$ may be unsafe under \textit{TSO} if the following hold
        \begin{tasks}(2)
            \task $w_{x}, w_{y}, r_{x} \notin st^{-}$.
            \task $(w_{x}, r_{x}), (r_{x}, w_{x}) \notin \po^{+}$.
            \task $(w_{x}, w_{y}), (w_{y}, w_{x}) \notin \po^{+}$.
        \end{tasks}
    \end{constraint}

    Constraints~\ref{cons:c0},\ref{cons:c1} together is still not precise enough to categorize the effect we desire.
    This can be seen in the following variation of $P'$ in Fig~\ref{conc:c3}.
    \begin{figure*}[!htbp]
        \centering
        \includegraphics[scale=0.6]{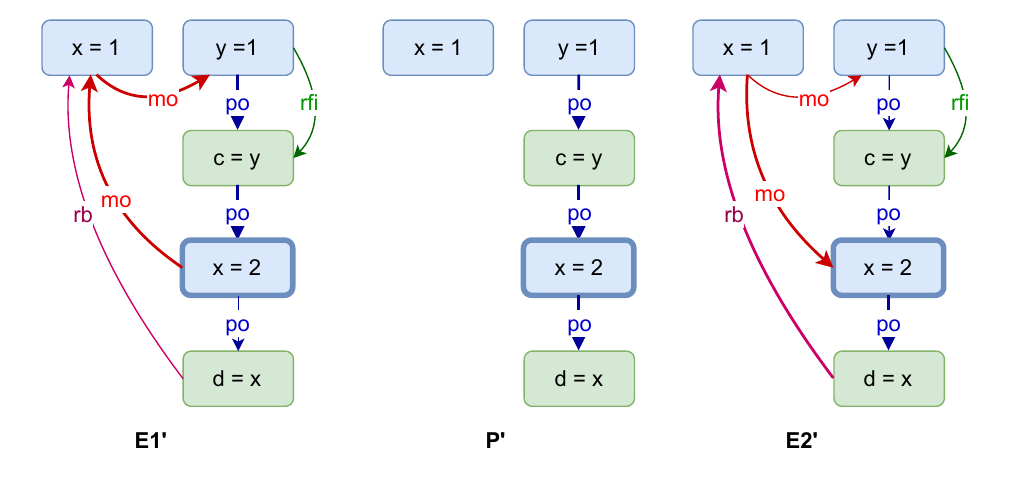}    
        \caption{A variation of effect $tr1$ which is safe under \textit{TSO}.}            
        \label{conc:c3}
    \end{figure*}
    Since $x=2$ is in between $y=1$ and $a=x$, the execution is inconsistent under \textit{TSO}.
    We either have $\mo;\hb$ cycle ($E1'$) or $\rb;\hb$ cycle ($E2'$), both of which violate the constraints of \textit{TSO}.
    Similar is the case if the write $x=2$ is replaced by a read-modify-write $u$ instead, giving us $\rb;\mo;[u];\po$ cycle for $E2'$.
    This leaves us with the effect $tr1$ being safe under \textit{TSO} \footnotemark.
    \footnotetext{
        It is easy to see the other candidate executions of $P$ and $P'$ support this claim.
    }
    Thus, an additional constraint is needed to ensure the effect should not cause any $w'_{x}(P) \cup u(P)$ to be in syntactic order between $w_{y}$ and $r_{x}$.
    %These constraints are not appropriate to precisely categorize what we really need.
    %Perhaps defining the effect purely in terms of its components and then later defining constraints on the source pre-trace in the actual definition might work.
    \begin{constraint}
        \label{cons:c2}
        The effect $P \mapsto P'$ may be under \textit{TSO} if the following hold 
        \begin{tasks}(2)
            \task* $\forall e \in (w_{x}(P) \cup u(P)) \ . \ (w_{y}, e), (e, r_{x}) \in \po(P) \implies (e, r_{x}) \in \po^{-}$.
            \task* $\nexists e \in (w_{x}(P) \cup u(P) \cup st^{+}) \ . \ (w_{y}, e), (e, r_{x}) \in \po^{+}$
        \end{tasks} 
    \end{constraint}

    We finally have our constraint on effects, which we combine to define the effects unsafe under \textit{TSO} as per program context. 
    \begin{definition}
        \label{def:tso-effect-const}
        Then a transformation-effect $P \mapsto P'$ involves \emph{tso unsafe write-read inlining} $\tuwri$ if it satisfies Constraints~\ref{cons:c0},\ref{cons:c1},\ref{cons:c2} for $P$ with
        \begin{tasks}(2)
            \task $(w_{y}, r_{y}), (r_{y}, w_{y}) \notin \po(P)$. %to make an rfe edge
            \task $(w_{y}, w_{x}), (w_{x}, w_{y}) \notin \po(P)$.
            \task $(w_{x}, r_{x}), (r_{x}, w_{x}) \notin \po(P)$.
            \task $(w_{x}, r_{y}), (r_{y}, w_{x}) \notin \po(P)$.
            \task $(r_{y}, r_{x}) \in \po(P)$.
        \end{tasks}
        or Constraints~\ref{cons:c1},\ref{cons:c2} for $P$ with
        \begin{tasks}(2)
            \task $(w_{y}, r_{x}) \in \po(P')$.
            \task $(w_{x}, w_{y}), (w_{y}, w_{x}) \notin \po(P')$.
            \task $(w_{x}, r_{x}), (r_{x}, w_{x}) \notin \po(P')$.
            %\task $\exists e \in (w_{x}(P') \cup rmw(P')) \ . \ (w_{y}, e), (e, r_{x}) \in \po(P')$.
        \end{tasks}
        
    \end{definition}

%% file: 3.Concrete/tso.tex
\subsection{Main Result}

    \label{subsec:main-result}
    
    We now formally state our result, followed by explaining the steps taken to prove it.
    The full proof is given in Appendix~\ref{sec:tso-proofs}.
    \begin{restatable}{theorem}{thmtsosccomp}
        \label{thm:tso-sc-comp}
        For transformation-effects $P \mapsto_{tr} P'$ not involving $\we$, $\wi$ and $\tuwri$, we have $\comp{TSO}{SC}$.
    \end{restatable}

    \subsubsection{Identifying Cycles in $E$, $E'$}

        The proof involves relying on the following lemma, which lists the minimal cycles that must exist to violate any constraint of \textit{SC} (proof in (Appendix~\ref{sec:sc-proofs})).
        \begin{restatable}{lemma}{lemscincons}
            \label{lem:sc-incons-exec}
            For any execution $E$ inconsistent under \textit{SC}, at least one of the following is true.
            \begin{tasks}(2)
                \task $\mo$ non-strict.
                \task $\mo;\po$ cycle.
                \task $\rfi;\po$ cycle.
                \task $\rb;\po$ cycle.
                \task $\rb;\rfe;\po$ cycle.
                \task $\mo;\rfe;\po$ cycle.
                \task $\rb;\hb;\rfe;\po$ cycle.
                \task $\rb;\mo$ cycle.
                \task $\mo;\rf$ cycle.
                \task $\rb;\mo_{ext};\rfe;\po$ cycle.
                \task $\rb;\mo;\po$ cycle.
            \end{tasks}
        \end{restatable}

    %\subsubsection{Addressing Read-Elimination Effects}    
    %    
    %    We first note that if any crucial set of $E$ involves eliminated reads, it does not matter; we only need to show that a crucial set $cr'$ of $E'$ is not one for $E$.
    %    Moreover, if a crucial set $cr$ of $E$ has only eliminated reads, then we can, by $pwc(SC)$, obtain a similar execution to $E'$ consistent under \textit{TSO}, which contradicts the premise of Lemma~\ref{lem:comp-gen-cond}. 
    %    Thus, WLOG we can assume no read-elimination effect is part of $tr$, i.e. $st^{-} \cap r(P) = \phi$.

        We proceed to use Lemma~\ref{lem:comp-gen-cond} where $B=SC$ and $M=TSO$ to infer the following about $E'$. 
        \begin{proposition}
            \label{prop:tso-cons-sc-incons}
            Any execution $E'$ consistent under \textit{TSO} but inconsistent under \textit{SC} has $\rb;\mo_{ext};\po$ cycle.
        \end{proposition}
        Next, since $E$ must be inconsistent under \textit{TSO}, every cycle except a subset of (k) from Lemma~\ref{lem:sc-incons-exec} can be true (Corollary~\ref{cor:tso-incons-exec}). 
        Specifically, (k), from Def~\ref{def:tso-model}, reduces to $\rb;\mo;[u];\po$ cycle in $E$.

        The proof now involves going case-wise over each cycle possible in $E$, comparing them with those in $E' \sim E$, and showing that it is possible to derive an $E'_{t} \in \llbracket P' \rrbracket_{SC}$ such that $\forall E_{t} \in \langle P \rangle \ . \ E_{t} \sim E'_{t} \implies E_{t} \in I\langle P \rangle_{SC}$.
        This simply implies $\neg \psf{SC}{tr}{P}$, thereby proving our result by contradiction.

    \subsubsection{Cases of Same Memory Deordering}
        For example Fig~\ref{conc:wr-deord} is a case where $[a=x];\rb;\po$ is a cycle in $E$, whereas $[a=x];\rb;\mo_{ext};\po$ is a cycle in $E'$.
        We can extract the information about $\po$ from both $E$ and $E'$, to identify an effect which we can instead prove to be unsafe under \textit{SC}.
        Since we have $st^{-} \cap w(P)$, we can infer that $P \mapsto_{tr} P'$ incurs same memory write-read de-ordering $(x=1, a=x) \in \po^{-}$.
        Such a de-ordering we show to be unsafe under \textit{SC} independent of the whole program context, thus, addressing such cycles in general. 
        \begin{figure*}[h]
            \centering
            \includegraphics[scale=0.6]{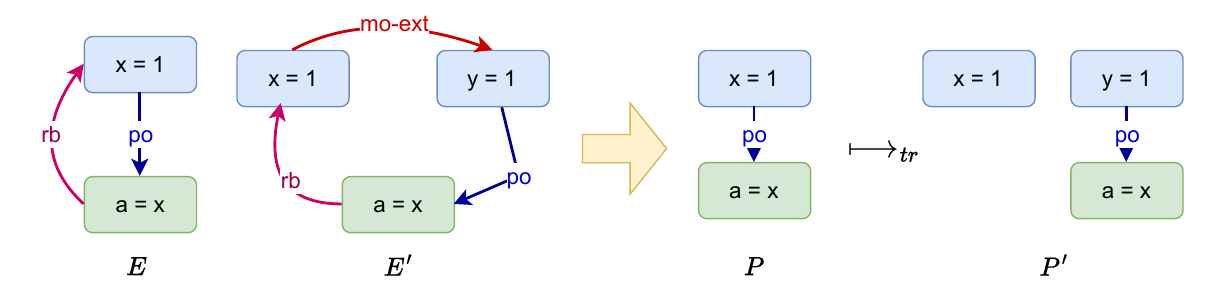}
            \caption{Executions $E$, $E'$ imply same memory write-read de-ordering effect (right).}
            \label{conc:wr-deord}
        \end{figure*}

    \subsubsection{Deriving Alternate $E$, $E'$}
        Other cases can be more involved, and may require extracting information on $\mo$ and $\rf$ in addition to $\po$ from $E$ and $E'$.
        For instance, take the case where $[a=x];\rb;\rfe;[b=x];\po$ cycle in $E$ as in Fig~\ref{conc:rr-deord} middle.
        If $P \mapsto_{tr} P'$ has $\{ a=x, b=x \} \in st^{-}$, then we can show there will always exist some $E_{t} \in \llbracket P \rrbracket_{TSO}$ such that $E_{t} \sim E'$.
        This violates our premise, thereby addressing all three cases\footnotemark.
        \footnotetext{
            These cases of read elimination can represent some form of load/store forwarding as well as unused read elimination done by the compiler. 
        }
        
        For the case where both reads exist in $P'$, we have same memory read-read de-ordering $(b=x, a=x) \in \po^{-}$.
        To prove this is unsafe under \textit{SC}, we use information from $\rf$ and $\mo$ relations in $E'$.
        We show it is always possible to manipulate $\mo(E')$ and $\rf(E')$ relations to derive some $E'_{t} \in \llbracket P' \rrbracket_{SC}$ for which there does not exist another $E_{t} \in \llbracket P \llbracket_{SC}$ such that $E_{t} \sim E'_{t}$.
        This by definition implies $\neg \psf{SC}{tr}{P}$\footnotemark.
        \footnotetext{
            De-ordering of same memory reads can be safe under \textit{SC} if there is no same memory concurrent write.
            However, such a case will also imply $\rb;\po$ cycle in $E'$, violating the premise.
        } 
        An example manipulating $\mo(E')$ is shown in Fig~\ref{conc:rr-deord} $E'$ (left) and corresponding $E'_{t}$ (right) derived by changing $[x=1];\mo_{ext};[y=1]$ to $[y=1];\mo_{ext};[x=1]$ instead. 
        We show that \emph{reversing} such $\mo_{ext}$ relations does not alter the cycle involving $a=x$ in $E$, yet removing it from $E'$.
        \begin{figure*}[h]
            \centering
            \includegraphics[scale=0.6]{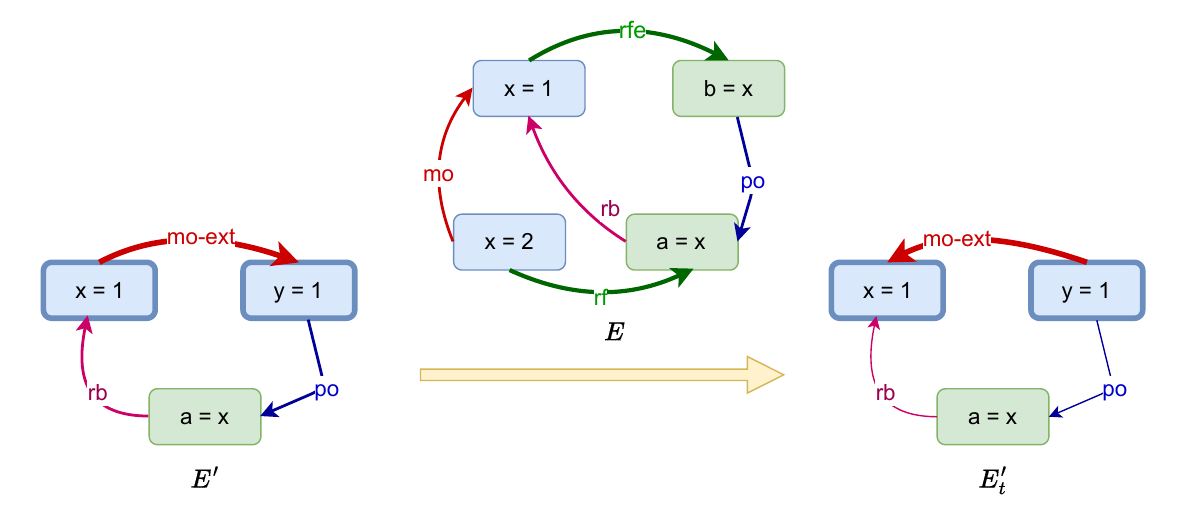}
            \caption{Execution $E$ (middle) implies same memory read-read de-ordering effect.}
            \label{conc:rr-deord}
        \end{figure*}
        
    \subsubsection{Usage of Constraint $\tuwri$}
        An example where the constraint Def~\ref{def:tso-effect-const} is required is shown in Fig~\ref{conc:iwri-ex}(a).
        We just have $[a=x];\rb;[x=2];\mo_{ext};[y=1];\rfe;[a=y];\po;[b=x]$ cycle in $E$ and $[a=x];\rb;[x=2];\mo_{ext};[w];\po$ cycle in $E'$.
        Using constraint $\tuwri$, we can infer $w$ is not $`y=1'$ in $P'$.
        This permits manipulating $[x=2];\mo;[w]$ to $[w];\mo;[x=2]$ to remove the cycle with $a=x$ in $E'$, while also preserving the cycle mentioned in $E$.
        This gives us our desired $E'_{t}$ and $E_{t}$.
        \begin{figure*}[h]
            \centering
            \includegraphics[scale=0.6]{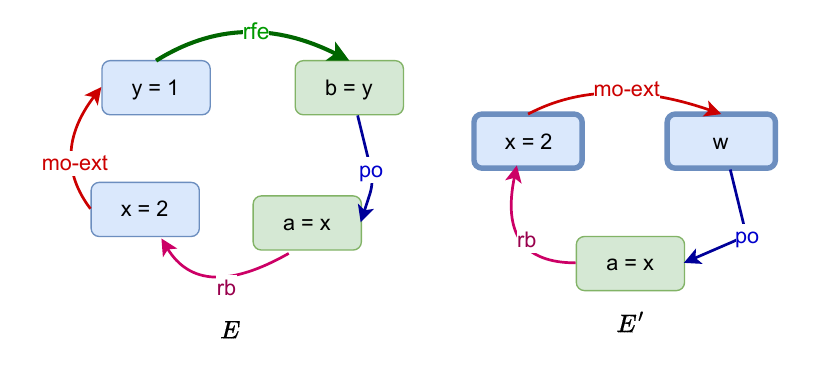}
            \caption{Example where the constraint $\tuwri$ is used to infer $\neg \psf{SC}{tr}{P}$.}
            \label{conc:iwri-ex}
        \end{figure*}

        For other cases, we show that in general manipulating $\mo_{ext}(E')$ and $\rf(E')$ relations will give similarly our desired $E'_{t}$ as described above. 
        The full proof of Theorem~\ref{thm:tso-sc-comp} is provided as supplementary material (cite Appendix).

%% file: 3.Concrete/races.tex
\subsection{Correlation with Races}

    \label{subsec:race-sra}
    Typically, safety of program optimizations are with the assumption that the source program is data-race-free, a semantic constraint that must be verified given source program. 
    To add, optimizations are considered safe only if they do not introduce data-races.
    This prohibits optimizing programs unless we can verify data-race-freedom, something which is very hard to determine, let alone program.
    Our result w.r.t. TSO shows that such a constraint may be too strict, and that it is possible to have constraints weaker than data-race-freedom, at the level of syntactic constraints over optimizations instead.

    An interesting correlation exists between our constraint and a subset of race inducing optimizations.
    To understand this, we refer to the example in Fig~\ref{conc:tr-ex}.
    \begin{figure*}[!htbp]
        \centering
        \begin{equation*}
            \inarrIII{
                x=1;
            }{
                y=1; \\ 
                a=x; \ (0)
            }{
                b=x; \ (1) \\ 
                c=y; \ (0)
            }
        \end{equation*}
        \caption{Example of a triangular race. Adapted from \cite{OwensTrftso}}
        \label{conc:tr-ex}
    \end{figure*}
    Observe that, under \textit{TSO} or \textit{SC} semantics, $b=1$ and $c=0$ implies $x=1$ must be visible to all processors before $y=1$.
    This would imply, observing $a=0$ would not be possible under \textit{SC}.   
    However, for \textit{TSO}, it is an acceptable outcome. 
    If $y=1$ remains in the write buffer, the read to $x$ can be done before $x=1$ is flushed to main memory. 
    Such an execution is called to exhibit a \emph{triangular race} (TR), a data race that occurs between $x=1$ and $a=x$ along with a preceding write $y=1$ to it \cite{OwensTrftso}.  
    
    Adapting the axiomatic definition from \cite{OwensTrftso}, we define triangular race as follows
    \begin{definition}
        \label{def:axiom-tr}
        A candidate execution $E$ has an \emph{Axiomatic TR} if there exists events $w_{x}, r_{x}, w_{y}$ such that 
        \begin{tasks}(2)
            \task $x \neq y$.
            \task $tid(w_{x}) \neq tid(r_{x})$.
            \task $tid(w_{y}) = tid(r_{x})$.
            \task $[w_{y}];\po;[r_{x}]$.
            \task $[w_{x}];\mo;[w_{y}]$.
            \task $\exists w'_{x} \neq w_{x} \ . \ [w'_{x}];\rf;[r_{x}]$.
            \task* $\forall w'_{x} \in st(p(E)) \ . \ [w'_{x}];\po;[w_{y}] \implies [w'_{x}];\mo;[w_{x}]$.
        \end{tasks}
    \end{definition}
    
    Def~\ref{def:axiom-tr} constraints (a), (b), (c), (d) gives us a shape of pre-traces whose executions can exhibit an Axiomatic TR.
    We can further refine this shape, using the following lemma below. 
    \begin{restatable}{lemma}{lempretracetoatr}
        \label{lem:pre-trace-to-atr}
        Any execution $E$ consistent under \textit{TSO} that has an Axiomatic TR with events $w_{x}, r_{x}, w_{y} \in st(P)$ will have the following additional constraint on $p(E)$.     
        \begin{tasks}
            \task $\nexists w'_{x} \ . \ [w_{y}];\po;[w'_{x}];\po;[r_{x}] \in \po(P)$.
            \task $\nexists u(z, a, v2) \in u(P) \ . \ [w_{y}];\po;[u];\po;[r_{x}] \in \po(P)$.
        \end{tasks} 
    \end{restatable}

    Finally, we can use the constraint on pre-traces to link triangular races to our constraint over transformation-effects portable from \textit{SC} to \textit{TSO}
    \begin{restatable}{lemma}{lemtriwri}
        \label{lem:tr-iwri}
        If $P \mapsto_{tuwri} P'$ then 
            $\exists E' \in \llbracket P' \rrbracket_{TSO} \ . \ E' \ \text{has an Axiomatic TR}$.
    \end{restatable}
    The proof is straightforward: the constraints over $P'$ can easily be implied using the Constraints~\ref{cons:c0},\ref{cons:c1},\ref{cons:c2} that constitute the effect $\tuwri$.

    Lemma~\ref{lem:tr-iwri} allow us to finally claim the following from Theorem~\ref{thm:tso-sc-comp}, correlating triangular races to our result\footnotemark.  
    \begin{corollary}
        \label{cor:tso-tr}
        For $P \mapsto_{tr} P'$ with no $\we$, $\wi$ and introducing an Axiomatic TR, we have $\comp{SC}{TSO}$.
    \end{corollary}

    \footnotetext{Notice that our original constraint is weaker, in that certain triangular races are permitted to be introduced.} 

%% file: 3.Concrete/sra.tex
\subsection{From Triangular Races to Causal Consistency}
    Just like $\tuwri$ relates to triangular races, an interesting correlation exists with Causal Consistency.
    We discovered that $\tuwri$ had an interesting correlation with Strong Release Acquire (SRA).
    SRA is known to be strictly weaker than TSO, and for two threaded programs with no updates/fences, the behaviors permitted by both models coincide \cite{nickSRA}.
    We instead provide an alternative description, one that establishes a relation with optimizations and races. 
    First, SRA can be viewed equivalent to allowing optimizations of Def~\ref{def:tso-effect-const} over TSO.
    The following theorem shows that in the absence of two specific constraints present in \textit{TSO}, the constraint $\tuwri$ is not required. 
    \begin{restatable}{theorem}{thmsrasccomp}
        \label{thm:sra-sc-comp}
        Considering constraints $a1 = \rb;\mo;\rfe;\po \ \text{irreflexive}$ and $a2 = \rb;\mo;[u];\po \ \text{irreflexive}$ we have $\comp{TSO \setminus \{ a1 \cup a2 \} }{SC}$ for $tr$ not involving $\wi$ and $\we$.
    \end{restatable}
    $TSO \setminus \{ a1 \cup a2 \}$, to our surprise, turns out to be precisely SRA.
    \begin{restatable}{lemma}{sra2}
        The memory model $M$ is equivalent to strong release acquire $SRA$.
    \end{restatable}

    \begin{proof}
        The equivalence follows directly from the alternative formulation of SRA in Def~10, Section 3.2 of \cite{nickSRA}.
    \end{proof}
    Since $\tuwri$ relates to triangular races (Lemma~\ref{lem:tr-iwri}), SRA can instead be described as the model that permits optimizations introducing triangular races over \textit{TSO}.
    Further, for programs which are write-write race free, we also know that behaviors permitted by both SRA and Release Acquire (RA)\footnotemark coincide \cite{nickSRA}. 
    This implies, for WW-race free programs, Theorem~\ref{thm:sra-sc-comp} holds even when $M=RA$.
    
    \footnotetext{RA forms a significant part of concurrency models used for programming languages like C++, Rust, Java, etc.}

%% file: 4.Disc/disc_intro.tex
\section{Discussion}

    %Main result over optimizations in general.
    Theorem~\ref{thm:tso-sc-comp} places 3 global constraints on possible syntactic changes involving shared memory events.
    More generally, any optimization designed relying on SC is applicable for TSO provided they adhere to the 3 syntactic constraints specified.
    To recap, these constraints are the prohibition of write introduction, write elimination and the introduction of particular \emph{triangular-races} in the optimized program. 
    These constraints are not too strict, incorporating thread-local (excluding write-elimination) as well as global optimizations like those of Fig~\ref{intro:ex1}.
    Equivalently, they are also sound, by accurately prohibiting the optimization under program contexts as those in Fig~\ref{intro:ex2}.
    Out of these 3 constraints, the exclusion of write elimination $\we$ is conservative; having such an assumption enables proving our results much easier.
    For now, we can port optimizations involving redundant write-before-write elimination effects; they are safe for TSO irrespective of program context \cite{MoiseenkoP}.  
    An immediate future work would be to include other forms of write elimination.

    \paragraph{Porting Program Analyses}
    Aggressive compiler optimizations are almost always backed by some program (data flow) analyses.
    Although several have been designed for concurrent programs, they rely on sequentially consistent semantics \cite{LeePad,ShashaSnir}, and may be unsound for weaker models like TSO \cite{Alglave2}.
    While this may require significant changes to existing analyses, our result can help ameliorate this problem to some extent.
    For instance, the \emph{delay set analysis} \cite{ShashaSnir} used to enable safe code motion under SC can also be used for TSO.
    This is because our constraint over effects do not prohibit any form of thread-local reordering effects, thereby enabling their portability across TSO.
    %Reordering update events exist, but they are anyways prohibited from being moved away.
    The same can be said for analyses that enable redundant/dead-code eliminations (provided they are of read or a form of write-before-write as stated).
    However, for analyses that are used to optimize programs beyond such syntactic changes, our result gives no guarantee.
    An example for this is the \emph{octagon range analysis}, a form of range analysis that can be used to identify redundant control dependencies in the program \cite{Mine}.
    Such analyses however, can be modified to be used for weaker models like TSO, SRA, etc. \cite{Alglave2}.
    
    \input{4.Disc/related.tex}

    \input{4.Disc/conclusion.tex}

%% file: 4.Disc/related.tex
\subsection{Related Work}
    \label{sec:related}
    %Memory Consistency and Optimizations
    The impact of memory consistency on program transformations can be traced back to when hardware optimizations were being introduced \cite{AdveS,OwensS,POWERSarkar}.
    However, compiler optimizations are much more complex and varied than those performed by hardware \cite{Dragon,LeePad,ShashaSnir}, and one of the earliest known impacts on them was seen designing memory model for Java \cite{PughW2}.
    Optimizations like bytecode reordering, copy propagation, thread-inlining, redundant read/write eliminations, etc. performed on Java Programs violated the specifications \cite{PughW,MansonP,SevcikJ}.
    C11 faces a similar problem: optimizations involving non-atomics, eliminating redundancies, strengthening, thread-inlining, etc. were unsafe \cite{MorissetR,VafeiadisV}.
    Not to mention that permitting optimizations also birthed the famed out-of-thin-air problem among language memory models \cite{BattyM}.

    %Proving optimizations safe
    The general problem of identifying which compiler optimizations are allowed by a memory model can be resolved conservatively by identifying which thread-local optimizations were safe under any program context \cite{MoiseenkoP,VafeiadisV,SafeOptSevcik}.
    These can be broadly categorized as adjacent reordering/elimination and introduction of redundant memory accesses, which can be inferred using trace semantic guarantees \cite{SafeOptSevcik}.  
    However, designing/verifying optimizations using these sound fragments for given language still requires reasoning with the associated weak memory model \cite{MikeDo,MetaDataSoham}.
    Some progress has been made in this direction, showing when it is adequate to rely on simple sequential reasoning to design optimizations for complex memory models \cite{MinkiC}.
    However, they are primarily for non-atomic optimizations that can be divided into these sound fragments. 
    
    %Proving Analyses safe
    Optimizations can also be performed using context-specific information obtained from a varied analysis on a multi-threaded program \cite{MartinR,ShashaSnir,LeePad,Alglave2}.
    A direction towards incorporating such optimizations can be to identify if analyses proven safe under a model be reused in another \cite{Alglave2}.
    Our work and the ones before us are exactly in this direction, albeit in a more general sense \cite{GopalAksh,gopaltransf}.
    Safety of transformation-effects over pre-traces are context specific, and proving Complete between two models allow us to gain context-specific transformations from one model to the other.
    The added advantage is that such an approach also encompasses all the sound (context-free) optimizations. 
    For instance, \cite{nickSRA} show that all thread-local optimizations sound in $RA$ are preserved in $SRA$.
    Our result on $\comp{SC}{TSO}$ also provides the same conclusion for \textit{SC} and \textit{TSO}, albeit in the opposite direction. 

%% file: 4.Disc/conclusion.tex
\section{Conclusion}

    In this paper, we identify syntactic constraints that enable porting SC optimizations across TSO.
    These also include optimizations leveraging concurrency, which may not be thread-local.
    We identify the correlation of our constraint with triangular races, followed by identifying syntactic constraints to port across SRA, a causally consistent memory model.
    Future work involves porting other variants of write-eliminations, as well as porting across Release Acquire, a significant subset of concurrent language models (C++20, Java) used today.

%% file: 5.Appendix/appendix_intro.tex
\appendix

\newpage

\input{5.Appendix/auxillary_appendix.tex}

\newpage

\input{5.Appendix/general_appendix.tex}

\newpage

\input{5.Appendix/sc_appendix.tex}

\newpage 

\input{5.Appendix/sra_appendix.tex}

\newpage

\input{5.Appendix/tso_appendix.tex}

%% file: 5.Appendix/auxillary_appendix.tex
\section{Auxiliary Elements from \cite{gopaltransf}}
  
    We restate here definitions and observations used from our prior work which are used for the proofs relevant to this paper.
    %Crucial Set 
    \begin{definition}
        \label{def:cr-rf}
        Given a execution $E$ inconsistent under memory model $M$, a set $cr \subseteq r(p(E))$ is a \emph{crucial set} if the execution $E' \sim E$ such that 
        \begin{tasks}(2)
            \task $\rf(E') = \rf(E) \setminus \ \{ [a];\rf;[b] \ | \ b \in cr \}$. 
            \task $\mo(E') = \mo(E)$.
        \end{tasks} 
        is consistent under $M$.
    \end{definition} 

    %Cra for each cycle
    \begin{definition}
        \label{def:crucial-refl-comp}
        We define $cra$ as a function that takes in a binary relation $s$ between memory events in $E$ and returns a set of reads $r$ such that without each of its associated $\rf$, the relation is irreflexive.
        \begin{align*}
            \forall r \in cra(s) \ . \ r \notin st(p(E)) \implies s \ \text{irreflexive}
        \end{align*}
    \end{definition}

    %Piecewise consistency
    \begin{definition}
        \label{def:piecewise-cons}
        A memory model $M$ is \emph{piecewise consistent} ($\pwc{M}$), if for any execution $E$ such that 
        \begin{tasks}(2)
            \task $\neg \wf{E}$.
            \task $c_{M}(E)$.  
            \task $\rf^{-1}$ functional.
            \task $\mo$ total order.
            \task $\forall loc \ . \ max(loc) \neq \phi$.
        \end{tasks}
        we have a candidate execution $E_{t}$ such that 
        \begin{tasks}(2)
            \task $c_{M}(E_{t})$.
            \task $p(E) = p(E_{t})$.
            \task $\mo(E) = \mo(E_{t})$.
            \task $\rf(E) \subset \rf(E_{t})$.
        \end{tasks}
    \end{definition}

    %Relation preserving models 
    \begin{definition}
        \label{def:bin-rel-presv}
        A memory model $M$ is \emph{relation-preserving} if any binary relation $\rel$ of execution $E$ required for the constraints of $M$ is preserved in every $E'$ ($\rel(E) \subseteq \rel(E')$) such that
        \begin{tasks}(3)
            \task $p(E) = p(E')$.
            \task $\rf(E) \subseteq \rf(E')$.
            \task $\mo(E) \subseteq \mo(E')$. 
        \end{tasks}
    \end{definition}

    %From inconsistent to consistent 
    \begin{lemma}
        \label{lem:cr-cons}
        For a given a candidate execution $E$ and memory model $M$ such that $\pwc{M}$ and $\neg c_{M}(E)$, if a crucial set $cr$ is non-empty, then there exists a candidate execution $E_{t}$ such that 
        \begin{tasks}(2)
            \task $c_{M}(E_{t})$.
            \task $p(E_{t}) = p(E)$. 
            \task $\po(p(E_{t})) = \po(p(E))$.
            \task $\mo(E_{t}) = \mo(E)$. 
            \task $\rf(E_{t}) \cap \rf(E) = \rf(E) \setminus cr$.  
        \end{tasks} 
    \end{lemma}

    %Using crucial sets for safety 
    \begin{lemma}
        %Refine the statement of the lemma and resolve the proper use of symbols for crucial sets 
        \label{lem:crucial-based-unsafety}
        Consider a memory model $M$ such that $\pwc{M}$.
        Further, consider a pre-trace $P$ and a transformation-effect $P \mapsto_{tr} P'$ that involves no elimination ($st^{-} = \phi$), along with two candidate executions $E \in I\langle P \rangle_{M}, E' \in I\langle P' \rangle_{M}$ such that $E \sim E'$.
        Then, 
        %If E does not have a crucial set but E' has, then we are done. 
        %If E has but is not a subset of E', then also we are done
        \begin{align*}
            \exists cr' \in Cr(E', M) \wedge Cr(E, M) = \phi \implies \neg \psf{M}{tr}{P}. \\
            \exists cr' \in Cr(E', M), cr \in Cr(E, M) \ \textit{s.t.} \ cr \nsubseteq cr' \implies \neg \psf{M}{tr}{P}. 
        \end{align*}
    \end{lemma}

    %Cra for SC cycles 
    \begin{restatable}{lemma}{lemsccra}
        \label{lem:sc-cra}
        For a given execution $E$, we have 
        \begin{align*}
            &cra(\mo) = \phi \\
            &cra(\hb) = \{r \ | \ ([e];\rfe;[r];\hb;[e] \neq \phi) \vee ([e];\rfi;[r];\po;[e] \neq \phi) \} \\ 
            &cra(\mo;\hb) = \{r \ | \ [e];\mo;\hb^{?};\rfe;[r];\hb;[e] \neq \phi \} \\
            &cra(\rb;\hb) = \{r \ | \  ([e];\rb;\hb^{?};\rfe;[r];\hb;[e] \neq \phi) \vee ([r];\rb;\hb;[r] \neq \phi) \} \\
            &cra(\rb;\mo;\hb) = \{r \ | \ ([e];\rb;\mo;\hb^{?};\rfe;[r];\hb;[e] \neq \phi) \vee ([r];\rb;\mo;\hb;[r] \neq \phi) \}
        \end{align*}  
    \end{restatable}

    %SC is piecewise consistent
    \begin{theorem}
        \label{thm:sc-piecewise-con}
        \textit{SC} is piecewise consistent - $pwc(SC)$.
    \end{theorem}

    \begin{restatable}{lemma}{lemsccr}
        \label{lem:sc-crucial-exist}
        Given an inconsistent execution $E$ under $SC$ such that $\mo$ is a strict total order such that $\mo;\po$ is irreflexive, a crucial set exists.
    \end{restatable}

%% file: 5.Appendix/general_appendix.tex
\section{Proving Generic Properties on Candidate Executions of Pre-traces}

    \label{sec:general-proofs}
    
    \lemcompgen*
    \begin{proof}
        
        From given, we have 
        \begin{flalign*}
            &E' \in \llbracket P' \rrbracket_{M} \wedge \nexists E \in \llbracket P \rrbracket_{M} \ . \ E \sim E'. \\ 
            &\to \llbracket P' \rrbracket_{M} \sqsubseteq \langle P \rangle. \tag*{$\langle P' \rangle \sqsubseteq \langle P \rangle$} \\ 
            &\to \forall E \sim E' \ . \ E \in I\langle P \rangle_{M}. \tag*{$\langle P \rangle = \llbracket P \rrbracket_{-} \cup I\langle P \rangle_{-}$} \\ 
            &\to E \in I\langle P \rangle_{B}. \tag*{($\weak{M}{B}$)} \\ 
            &\to E' \in I\langle P' \rangle_{B}. \tag*{($\llbracket P' \rrbracket_{B} \sqsubseteq \llbracket P \rrbracket_{B}$)} &
        \end{flalign*}

        Hence proved.
        
    \end{proof}
    
    %Adding $\mo$ relations is always possible such that it is strict not conflicting $\po$.
    \begin{lemma}
        \label{lem:sc-rr:welim-p0}
        Consider a candidate-execution $E$ of $P$ such that $\mo(E) = \phi$.
        %\begin{align*}
        %    &[e1];\po;[e2] \in \po(p(E)) \implies [e1];\mo;[e2] \in \mo(p(E)). \\
        %    &tid(e1) \neq tid(e2) \implies [e1];\mo;[e2] \notin \mo(E) \wedge [e2];\mo;[e1] \notin \mo(E). \\ 
        %    &\forall e1 \in r(p(E)) \exists e2 \in w(p(E)) \ . \ [e2];\rf;[e1]. 
        %\end{align*}
        Then, there exists a well-formed execution $E' \in \langle P \rangle$ such that
        \begin{align*}
            &\mo \ \text{strict total order} \ \wedge \ \mo;\po \ \text{irreflexive}.
        \end{align*}     
    \end{lemma}
    
    \begin{proof}
        
        We first note that since $\po$ is a strict total-order per thread, we can first add $\mo$ relations consistent to $\po$.
        This keeps $\mo(E')$ as a strict total-order per-thread, with $\mo;\po$ irreflexive.
        
        To finish, we need to add the remaining $\mo$ relations, i.e., those between writes of different threads. 
        We first add $\mo$ relations from initialization writes to writes of each thread as follows:
        \begin{align*}
            tid(w1) \neq tid(w2) \wedge tid(w1) = 0 \implies [w1];\mo;[w2].
        \end{align*} 
        This by construction keeps the resultant $\mo$ as strict a partial order.
        
        We add the rest by induction on number of $\mo$ relations $n$ added, showing that this incremental addition preserves the strict partial order of resultant $\mo$. 
        \begin{itemize}
            \item Base case: $n=0$ - By definition, the order that exists is strict and partial. 
            \item Inductive case $n=k$: If current $\mo$ is strict and partial after adding $k$ edges, then $\mo$ is strict and partial after adding $k+1$ edges. 
            \begin{itemize}
                \item Consider the next pair of writes $w1, w2$ between which $\mo$ relation is to be added, such that $tid(w1) \neq tid(w2)$.
                \item If either $[w1];\mo;[w2]$ or $[w2];\mo;[w1]$ does not make $\mo$ non-strict, we are done. 
                \item If $[w1];\mo;[w2]$ makes $\mo$ non-strict, add $[w2];\mo;[w1]$ instead, and vice versa.
                \item If both cause $\mo$ to be non-strict then, 
                \begin{itemize}
                    \item There exists some $w3$ such that $[w1];\mo;[w2];\mo;[w3];\mo$ forms a cycle.
                    \item There exists some $w4$ such that $[w2];\mo;[w1];\mo;[w4];\mo$ forms a cycle.
                    \item This implies $[w1];\mo;[w4];\mo;[w2];\mo;[w3];\mo$ forms a cycle.
                    \item Thus, the $\mo$ without adding $(k+1)^{th}$ relation is non-strict.
                    \item This contradicts our inductive assumption for $n=k$. 
                \end{itemize}
            \end{itemize}
        \end{itemize}
    
        Since candidate-executions consist of finite events, we eventually obtain $E'$ such that $\mo(E')$ is a strict total-order, not conflicting $\po$.
        Hence proved.
    
    \end{proof}

    \begin{corollary}
        \label{cor:sc-rr:welim-p1}
        If $\mo(E)$ is instead a strict partial order not conflicting $\po$, then there exists a well-formed execution $E' \in \langle P \rangle$ such that
        \begin{align*}
            &\mo \ \text{strict total order} \ \wedge \ \mo;\po \ \text{irreflexive}.
        \end{align*}  
    \end{corollary}

    \begin{proof}
        Inductive case of the proof of Lemma~\ref{lem:sc-rr:welim-p0}. 
    \end{proof}

    \begin{restatable}{lemma}{lemflipmoext}
        \label{lem:sc-wr:p1}
        Consider $r_{wr}$ to be any cycle of the form $[r_{x}];\rb;\mo^{?}_{ext};\hb^{?}$ in a candidate execution $E'$ such that 
        \begin{tasks}[style=itemize](2)
            \task $\mo$ strict total order.
            \task $\mo^{?};\hb$ irreflexive.
            \task* $[w_{x}];\hb;[w'_{x}] \wedge [w_{x}];\rf;[r_{x}] \implies [r_{x}];\rb;[w'_{x}];\hb$ irreflexive.
            \task $[r_{x}];\rb;(\mo \setminus \mo_{ext})$ irreflexive.
            %\task $[r_{x}];\rb;\hb$ irreflexive.
            %\task $\rb;\mo$ irreflexive.
        \end{tasks}
        Then, there exists another execution $E'_{t}$ satisfying the above, such that   
        \begin{tasks}(3)
            \task $p(E') = p(E'_{t})$. 
            \task $\rf(E') = \rf(E'_{t})$.  
            %\task $\mo_{x}(E'_{t}) = \mo_{x}(E')$. 
            \task $r_{wr}$ irreflexive.
            %\task $\forall w'_{x} \ . \ [r_{x}];\rb;[w'_{x}];\mo_{ext};\hb$ cycle in $E'$.
            \task* $[w_{x}];\rf;[r_{x}] \in \rf(E') \implies \mo;[w_{x}](E') = \mo;[w_{x}](E'_{t})$.
        \end{tasks}
    \end{restatable}

    \begin{proof}

        We prove this by constructing $E'_{t}$ modifying $\mo(E')$. 
        
        From given, we have $r_{wr}$ cycle in $E'$.
        Let $w_{x}, w_{y}, w'_{x}$ be the writes in $E'$ such that 
        \begin{align*} 
            [r_{x}];\rf^{-1};[w_{x}];\mo_{x};([w'_{x}];\mo_{ext};[w_{y}])^{?};\hb^{?} \ \text{cycle}.
        \end{align*}
        We note that without $[w_{x}];\mo_{x};[w'_{x}];\mo_{ext}^{?};[w_{y}]$, the cycle will cease to exist.
        Thus, we change one of these relations: to either $[w_{y}];\mo_{ext};[w'_{x}]$ or $[w'_{x}];\mo_{x};[w_{x}]$, giving us $r_{wr}$ irreflexive.
        If $[w_{x}];\mo_{x};[w'_{x}] \notin \mo_{ext}(E')$, then we change the other relation.
        Changing all such $\mo$ relations will render $r_{wr}$ irreflexive, satisfying (c).
        
        The change in $\mo$ as stated, can leave us a relation $\mo$ that may not be strict and total.
        The cause for this is the new $\mo$ may not longer by antisymmetric, due to some cycle involving writes $w'_{x}, w_{y}$ or $w_{x}, w'_{x}$.
        Since $\mo$ is transitive, we only require another $w$ for this, giving us either 
        \begin{tasks}(2)
            \task $[w_{y}];\mo_{ext};[w'_{x}];\mo;[w];\mo$ cycle or
            \task $[w'_{x}];\mo_{ext};[w_{x}];\mo;[w];\mo$ cycle.
        \end{tasks}        
        We can, modify $\mo$ further to fix this.
        However, we must also ensure it does not result in $\mo$ conflicting $\hb$, satisfying the other 3 constraints in addition to $\mo$ being strict and total.
        We provide the argument for a cycle with $w'_{x}, w_{y}, w$, which is the same for other forms of cycles where $y=x$ 
        \begin{itemize}
            \item First, we note that both $[w'_{x}];\mo;[w]$ and $[w];\mo;[w_{y}]$ do not conflict $\hb$, due to given constraint on $E'$.  
            \item Second, at most one of the pairs $(w'_{x}, w)$ or $(w, w_{y})$ can have a $\hb$ relation. If both pairs do, then either $\mo;\hb$ cycle or $[w_{y}];\mo_{ext};[w'_{x}]$ is untrue due to transitive relation $\hb$. 
            \item Thus, at least one of the following must hold: $[w];\mo_{ext};[w_{y}]$ or $[w'_{x}];\mo_{ext};[w]$.
            \item Third, if $[w];\hb;[w_{y}]$, then we must have $[r_{x}];\rb;[w'_{x}];\mo_{ext};[w];\hb$ cycle.
            \item This violates our premise that we have $r_{wr}$ irreflexive.
            \item Thus, we will always have $[w];\mo_{ext};[w_{y}]$. 
            \item Removing all such $\mo_{ext}$ relations renders $\mo$ antisymmetric, leaving $\mo$ strict and total again.
            \item Since we only change $\mo_{ext}$ relations of $E'$, we have (a), (b), $\mo^{?};\hb$, $[r_{x}];\rb;\hb$ and $\rb;[\mo \setminus \mo_{ext}]$ irreflexive in $E'_{t}$. 
            \item Finally, since we only change at most relation of the form $[w_{x}];\mo;[w'_{x}]$, $\mo;[w_{x}](E') = \mo;[w_{x}](E'_{t})$.
        \end{itemize}

        The same argument follows for a cycle with $w_{x}, w'_{x}, w$.

        Hence, proved.
    \end{proof}

    \begin{corollary}
        \label{corr:mo-ext-flip}
        If $E'$ has $[r_{x}];\rb;\hb$ irreflexive and $[r_{x}];\rb;\mo$ irreflexive, then $\mo_{x}(E') = \mo_{x}(E'_{t})$.
    \end{corollary}

    \begin{proof}
        
        For the case of addressing only $w_{y}, w'_{x}, w$ cycles, $\mo_{x}(E')$ remains unchanged.

    \end{proof}

        %\corflipmoext*
        %
        %\begin{proof}
        %    
        %    Since $\mo_{x}(E') = \mo_{x}(E'_{t})$, and that $[r_{x}];\rb;\hb$ irreflexive in $E'$, we can infer it is also irreflexive in $E'_{t}$.
        %    %From $E'_{t}$ we have $[w'_{x}];\rf;[r_{x}]$.
        %    %Deriving $E'$ from $E'_{t}$ requires only changing $[w_{x}];\mo_{ext};[w_{y}]$ and any other $\mo$ relations of the form $[w_{x}];\mo;[w]$ or $[w];\mo_{ext};[w_{y}]$.
        %    %In $E'_{t}$ and $E$, we have $[w'_{x}];\mo_{loc};[w_{x}]$, thus it remains unchanged.
        %    %Since $[r_{x}];\rb;\hb$ irreflexive in $E'$, it is also irreflexive in $E'_{t}$.
        %    %Hence proved.
        %
        %\end{proof}

%% file: 5.Appendix/sc_appendix.tex
%RESULTS PERTAINING TO SC--------------------------------------------------------------------------------------------------
\section{Proofs Pertaining to Sequential Consistency}
    
    \label{sec:sc-proofs}
    %ALTERNATE MODEL FOR SC TO BE USED FOR SOME PROOFS 
    %Lemma~\ref{lem:sc-incons-exec}, along with the following information from \cite{gopaltransf} about $cra$ will help in reasoning about crucial sets of $E$ and $E'$ from Lemma~\ref{lem:comp-gen-cond}.
    %\begin{proposition}
    %    \label{prop:sc-cra}
    %    For a given execution $E$, we have 
    %    \begin{align*}
    %        &cra(\mo) = \phi. \\ 
    %        &cra(\mo;\po) = \phi. \\      
    %        &cra(\hb) = \{r | \rfe;[r];\hb \vee \rfi;[r];\po \}. \\ 
    %        &cra(\mo;\hb) = \{r | \mo;\hb^{?};\rfe;[r];\hb \vee \mo;\rf;[r] \}. \\ 
    %        &cra(\rb;\hb) = \{r | \rb;\hb^{?};\rfe;[r];\hb \vee [r];\rb;\hb \}. \\ 
    %        &cra(\rb;\mo;\hb) = \{r | \rb;\mo;\hb^{?};\rfe;[r];\hb \vee [r];\rb;\mo;\hb \}. \\ 
    %        &cra(\rb;\mo) = \{r | [r];\rb;\mo \}. 
    %    \end{align*}  
    %\end{proposition}
%
    %\pending{
    %    Change the above to reflect what it really means.
    %    Take it from the FAC version
    %}

    %Prop~\ref{prop:sc-cra}, with Lemma~\ref{lem:sc-incons-exec} hints that as long as $\mo$ does not conflict $\po$, a crucial set for $SC$ will exist.
    %This is indeed true, and proves useful to show any $E'$ from Lemma~\ref{lem:comp-gen-cond} for our models will have a crucial set.

    The following lemma related to $SC$ will be used to infer cycles in any inconsistent executions under $TSO$, $SRA$ and $RA$.
    \lemscincons*
    \begin{proof}
        
        We go case-wise on the cycles in $E$ violating rules of $SC$ as in Def~\ref{def:sc-model}.
        For each, we show that at least one of the aforementioned cycles must be true. 
        We do so by identifying compositions which form cycles with fewer events involved. 
        Since pre-traces are of finite events, every reduction will eventually lead to the desired compositions. 
        \begin{description}
            \item[Rule (a) violated:] $\mo$ non-strict - (a) is true.  
            \item[Rule (b) violated:] $\hb$ cycle.
            \begin{flalign*}
                &\hb \\ 
                &\to \rf;\po;\hb^{?} \\ 
                &\to \rfi;\po \vee \rfe;\po;\hb \\
                &\to (c) \vee [w];\rfe;\po;[w'];\hb \\ 
                &\to (c) \vee [w'];\mo;[w];\rfe;\po \vee [w];\mo;[w'];\hb \\ 
                &\to (c) \vee (f) \vee \mo;\hb.  
            \end{flalign*}  
            \item[Rule (c) violated:] $\mo;\hb$ cycle.
            \begin{flalign*}
                &\mo;\hb \\ 
                &\to \mo;\po \vee \mo;\rf \vee \mo;[w];\hb;[w'];\rf \vee \mo;[w];\hb;[e];\po  \\ 
                &\to (b) \vee (i) \vee [w'];\mo;[w];\hb \vee \mo;\rf \vee  \mo;[w];\hb;[w'];\po \vee \mo;[w];\hb^{?};[w'];\rf;[r];\po \\ 
                &\to (b) \vee (i) \vee \mo;\hb \vee (i) \vee \mo;[w];\rfe;\po \vee [w];\mo;[w'];\rfi;\po \vee \mo;[w];\hb;[w'];\rfe;\po \\
                &\to (b) \vee (i) \vee (f) \vee [w];\mo;[w'];\po \vee [w'];\rfi;\po \vee \mo;[w']\hb. \\
                &\to (b) \vee (i) \vee (c) \vee (f) \vee \mo;\hb.    
            \end{flalign*} 
            \item[Rule (d) violated:] $\rb;\hb$ cycle.  
            \begin{flalign*}
                &\rb;\hb  \\ 
                &\to \rb;\hb^{?};\po;[r_{x}] \\ 
                &\to \rb;\po;[r_{x}] \vee \rb;[w];\hb^{?};[w'];\rfe;\po;[r_{x}] \\ 
                &\to \rb;\po;[r_{x}] \vee \rb;[w];\rfe;\po;[r_{x}] \vee \rb;[w];\hb;[w'];\rfe;\po;[r_{x}] \\ 
                &\to (d) \vee (e) \vee \rb;[w];\hb;[w'];\rfe;\po_{x} \vee [w'];\mo;[w];\hb \\
                &\to (d) \vee (e) \vee (g) \vee \mo;\hb.   
            \end{flalign*}
            \item[Rule (e) violated:] $\rb;\mo$ cycle - (h) is true. 
            \item[Rule (f) violated:] $\rb;\mo;\hb$ cycle 
            \begin{flalign*}
                &\rb;\mo;\hb \\
                &\to \rb;[w];\mo;[w'];\hb \vee \rb;\mo_{ext};\hb \\ 
                &\to \rb;[w];\hb;[w];\hb \vee [w];\mo;[w'];\hb \vee \rb;\mo_{ext};\hb{?};\po \\ 
                &\to \rb;\hb \vee \mo;\hb \vee \rb;\mo_{ext};\po \vee \rb;\mo_{ext};\hb^{?};\rfe;\po \\ 
                &\to \rb;\hb \vee \mo;\hb \vee (k) \vee \rb;\mo_{ext};\rfe;\po \vee \rb;\mo_{ext};[w];\hb;[w'];\rfe;\po \\ 
                &\to \rb;\hb \vee \mo;\hb \vee (k) \vee (j) \vee \rb;\mo_{ext};[w'];\rfe;\po \vee [w'];\mo;[w];\hb \\ 
                &\to \rb;\hb \vee \mo;\hb \vee (k) \vee (j) \vee \mo;\hb. 
            \end{flalign*}
        \end{description} 

        Hence proven.

    \end{proof}

    %proving write-write, read-write and read-read deordering unsafe under SC
    \subsection{Proving Relevant De-ordering Effects Unsafe under $SC$}
        \label{subsec:sc-deord}
        We first show that general write-write de-ordering effect is unsafe under $SC$
        \begin{lemma}
            \label{lem:transf:ww-sep:sc}
            Consider a pre-trace $P$ and two writes $w1, w2$ such that 
            \begin{tasks}(2)
                \task $w1 \in st(P)$.
                \task $w2 \in st(P)$.
                \task $[w1];\po;[w2] \in \po(P)$.
            \end{tasks}  
            Then for any transformation-effect $P \mapsto_{tr} P'$ such that $[w1];\po;[w2] \in \po^{-}$ and $w1, w2 \notin st^{-}$, we have $\neg \psf{SC}{tr}{P}$.
        \end{lemma}

        \begin{proof}

            \begin{itemize}
                \item Any $E \in \langle P \rangle$ with $[w2];\mo;[w1]$ is inconsistent under $SC$ due to violation of Rule (c): $[w2];\mo;[w1];\po$ forms a cycle.
                \item Any execution $E' \in \langle P' \rangle$,  $[w2];\mo;[w1];\po$ is irreflexive.
                \item By Lemma~\ref{lem:sc-rr:welim-p0}, we always have an execution $E' \in \langle P' \rangle$ with $\mo$ strict and total, not conflicting $\po$.
                \item With $pwc(SC)$ and Lemma~\ref{lem:sc-crucial-exist}, we can obtain an execution $E' \in \llbracket P' \rrbracket_{SC}$.
                \item This gives us $\neg \psf{SC}{tr}{P}$. 
            \end{itemize}

        \end{proof}

        %WRITE-READ/READ-WRITE DEORDERING is unsafe under SC
        Next, we want to show same memory read-write or write-read deordering is unsafe under $SC$.
        For this, we first show if a read and write to same memory belong to different threads (other than initialization), then there always exists an consistent execution under $SC$ where the read value is from the respective write. 
        \begin{lemma}
            \label{lem:sc:p1}
            If for a given pre-trace $P$ having events $w_{x}$ and $r_{x}$, 
            \begin{align*}
                &tid(r_{x}) \neq tid(w_{x}) \wedge tid(w_{x}) \neq 0 \ \vee \\ 
                &[w_{x}];\po;[r_{x}] \wedge (\nexists e \ . \ m(e) = m(w_{x}) \wedge [w_{x}];\po;[e];\po;[r_{x}])
            \end{align*}
            then 
            \begin{align*}
                \exists E \in \llbracket P \rrbracket_{SC} \ . \ [w_{x}];\rf;[r_{x}] \in \rf(E).
            \end{align*}      
        \end{lemma}

        \begin{proof}

            We prove this by construction of execution $E$.
            We first add the relation $[w_{x}];\rf;[r_{x}]$.
            Now, to add $\mo$, we first add $\mo$ relations with $w_{x}$ using the following constraint.
            \begin{align*}
                \forall w \ . \      &[w];\po;[r_{x}] \implies (w, w_{x}) \in \mo(E). \\ 
                                      &[r_{x}];\po;[w] \implies (w_{x}, w) \in \mo(E). \\ 
                                      &[w];\po;[w_{x}] \implies (w, w_{x}) \in \mo(E). \\
                                      &[w_{x}];\po;[w] \implies (w_{x}, w) \in \mo(E). 
            \end{align*}
            
            This will ensure $[r_{x}];\rb;\po$ irreflexive and $\mo;\rfe;[r_{x}];\po$ irreflexive.  
            With the given condition between events $w_{x}$ and $r_{x}$, $\mo$ also does not conflict $\po$.
            Lastly, since only $[w_{x}];\mo$ is non-empty, $\mo$ is strict partial order.

            Next, we add more $\mo$ relations as follows 
            \begin{align*}
                \forall w, w'_{x} \ . \ [w_{x}];\mo_{x};[w'_{x}] \wedge [w];\po;[r_{x}] \implies (w, w'_{x}) \in \mo(E). 
            \end{align*}
            This ensures $[r_{x}];\rb;\mo;\po$ irreflexive.
            The resultant $\mo$ is still strict and partial not conflicting $\po$ (due to $[w];\mo;[w_{x}]$).

            Now, we can add the remaining $\mo$ relations as per Corollary~\ref{cor:sc-rr:welim-p1}, such that $\mo$ is strict and total, not conflicting $\po$.
            This may leave us with an inconsistent execution $E$.
            For this, using Lemma~\ref{lem:sc-crucial-exist}, we can infer a crucial set exists.
            Using Lemma~\ref{lem:sc-incons-exec}, Lemma~\ref{lem:sc-cra}, and the above irreflexive relations, we can infer a crucial set need not have $r_{x}$.
            From Lemma~\ref{lem:cr-cons}, we obtain our desired candidate execution $E$ consistent under $SC$. 

        \end{proof}

        Using Lemma~\ref{lem:sc:p1}, we can finally prove same memory read-write / write-read reordering unsafe under $SC$.
        \begin{lemma}
            \label{lem:transf:rw-wr-eq-sep:sc}
            Consider a pre-trace $P$ and transformation-effect $P \mapsto_{xdeord} P'$ such that there exists two $\po$ ordered read $r_{x}$ and write $w_{x}$ events in $P$ where  
            \begin{tasks}
                \task $[r_{x}];\po;[w_{x}] \in \po(P) \implies [r_{x}];\po;[w_{x}] \in \po^{-} \wedge r_{x}, w_{x} \notin st{-}$.
                \task $[w_{x}];\po;[r_{x}] \in \po(P) \implies [w_{x}];\po;[r_{x}] \in \po^{-} \wedge r_{x} \notin st^{-}$. 
            \end{tasks}
            Then $\neg \psf{SC}{tr}{P}$.
        \end{lemma}

        \begin{proof}

            The proof can be divided into two cases, one for each type of syntactic order between $r_{x}$ and $w_{x}$ in $P$.
            For each case, we first show that a class of executions $E$ with a specific $\rf$ and/or $\mo$ relation with $r_{x}$ is always inconsistent under $SC$.
            This is followed by showing that the can exist an execution with the same $\rf$ and/or $\mo$ relation that is consistent under $SC$ after the transformation-effect.  
            This would directly imply $\neg \psf{SC}{tr}{P}$.
            \begin{description}
                \item[Case a:] $[r_{x}];\po;[w_{x}] \in \po(P)$.
                    \begin{itemize}
                        \item Any $E \in \langle P \rangle$ with $[w_{x}];\rf;[r_{x}]$ is inconsistent under $SC$ due to violation of Rule (b): $[r_{x}];\po;[w_{x}];\rf$ forms a cycle.
                        %\item Considering any such $E$ with $\mo$ strict and total not conflicting $\po$, every crucial set $cr \in Cr(E, SC)$ would have $r_{x}$.
                        %\item By Lemma~\ref{lem:transf:ww-sep:sc}, we can, WLOG, assume no other write exists program ordered between $w_{x}$ and $r_{x}$ in $P'$.
                        \item If $[w_{x}];\po;[r_{x}] \notin \po^{-}$, by Lemma~\ref{lem:sc:p1}, we can infer there exists $E' \in \llbracket P' \rrbracket_{SC}$ with $[w_{x}];\rf;[r_{x}]$.    
                        \item Otherwise, based on the existence of $[w_{x}];\po;[w'_{x}];\po;[r_{x}]$ in $\po(P')$, we can either use Lemma~\ref{lem:transf:ww-sep:sc} or Lemma~\ref{lem:sc:p1} to infer a consistent execution exists.  
                    \end{itemize}
                \item[Case b.1:] $[w_{x}];\po;[r_{x}] \in \po(P) \ \wedge \ (\nexists e \in w_{x}(P) \ . \ [w_{x}];\po;[e];\po;[r_{x}] \in \po(P)$).
                    \begin{description}
                        \item We now go case-wise on whether another write to same memory exists program ordered before $r_{x}$ in $P'$.
                        \item[Subcase 1:] $\exists e \in w(P') \ . \ mem(e) = mem(r_{x}) \wedge [e];\po;[r_{x}] \in \po(P')$
                        \begin{itemize}
                            \item WLOG, we can assume there does not exist any other write to the same memory syntactically ordered between $e$ and $r_{x}$ in $P'$.
                            \item Any $E \in \langle P \rangle$ with $[e];\rf;[r_{x}]$ and $[e];\mo;[w_{x}]$ will be inconsistent under $SC$ as it violates Rule (d): $[r_{x}];\rb;[w_{x}];\po$ forms a cycle.
                            \item By Lemma~\ref{lem:sc:p1}, we can infer that there exists $E' \in \llbracket P' \rrbracket_{SC}$ with $[e];\rf;[r_{x}]$.     
                            \item To show that $E'$ can have $[e];\mo;[w_{x}] \in \mo(E')$, note that $[r_{x}];\rb;[w_{x}];\po$ no longer forms any cycle.
                            \item For any other write event $e'$ such that $[e];\po;[r_{x}]$, we can have $[e'];\mo;[w_{x}]$ to avoid making $[r_{x}];\rb;\mo;[e'];\po$ form a cycle.
                        \end{itemize}
                        \item[Subcase 2:] $\nexists e \in w(P') \ . \ mem(e) = mem(r_{x}) \wedge [e];\po;[r_{x}] \in \po(P')$
                        \begin{itemize}
                            \item Any $E \in \langle P \rangle$ with $[w_{init}];\rf;[r_{x}]$ is inconsistent under $SC$ due to violation of Rule (c) or (d): If $[w_{x}];\mo;[w_{init}]$, then $\mo;\po$ forms a cycle and if $w_{init};\mo;[w_{x}] \in \mo(E)$ then $[r_{x}];\rb;[w_{x}];\po$ forms a cycle. 
                            \item By Lemma~\ref{lem:sc:p1}, we can infer that there exists $E' \in \llbracket P' \rrbracket_{SC}$ with $[w_{init}];\rf;[r_{x}]$.     
                        \end{itemize}
                    \end{description}
                \item[Case b.2:] $[w_{x}];\po;[r_{x}] \in \po(P) \ \wedge \ (\exists e \in w_{x}(P) \ . \ [w_{x}];\po;[e];\po;[r_{x}])$
                    \begin{description}
                        \item With $[w_{x}];\po;[r_{x}] \in \po^{-}$, we have the following possibilities for $e$ in $P'$.
                        \item[Subcase 1:] $e \notin st^{-}$ 
                        \begin{itemize}
                            \item $[r_{x}];\po;[w_{x}] \in \po(P')$ - Since $\po$ is also a strict total order per-thread, we also have $[w_{x}];\po;[e] \in \po^{-}$. By Lemma~\ref{lem:transf:ww-sep:sc} we have $\neg \psf{SC}{tr}{P}$. 
                            \item $[r_{x}];\po;[w_{x}] \notin \po(P')$ - For this $tid(e) \neq tid(r_{x})$ reduces to Case 2 and $tid(e) \neq tid(w_{x})$ can be resolved by Lemma~\ref{lem:transf:ww-sep:sc}. 
                        \end{itemize}
                        \item[Subcase 2:] $e \in st^{-}$
                            Then we can first remove $e$ as an intermediate effect, and the remaining effect can be resolved by the above cases.
                    \end{description}
            \end{description}

            Hence, proved.

        \end{proof}

%% file: 5.Appendix/sra_appendix.tex
\section{Proving SRA Complete w.r.t. TSO and SC}
    
    \label{sec:sra-proofs}
    The following is inferred from Lemma~\ref{lem:sc-incons-exec}.
    \begin{restatable}{corollary}{corsraincons}
        \label{cor:sra-incons-exec}
        At least one of the following is true for any $E$ inconsistent under $SRA$
        \begin{tasks}(2)
            \task $\mo$ cycle.
            \task $\mo;\po$ cycle.
            \task $\rfi;\po$ cycle.
            \task $\rb;\po$ cycle.
            \task $\rb;\rfe;\po$ cycle.
            \task $\mo;\rfe;\po$ cycle.
            \task $\rb;\hb;\rfe;\po$ cycle.
            \task $\rb;\mo$ cycle.
            \task $\mo;\rf$ cycle.
        \end{tasks}
    \end{restatable}

    \begin{proof}
        Checking the rules of $SRA$ from Def~\ref{def:sra-model}, we know that $\rb;\mo;\hb$ need not be irreflexive. 
        Hence proved.
    \end{proof}

    %\begin{theorem}
    %    \label{thm:sc-wr:comp-noelim-nowrord}
    %    For transformation-effects not involving write-elimination $st^{-} \cap w(P) = \phi$ or $wri$, we have $\comp{SC_{WR}}{SC}$.
    %\end{theorem}

    \thmsrasccomp*

    \begin{proof}
        
        From Lemma~\ref{lem:comp-gen-cond} we know the exact $tr$ which requires proving unsafe under $SC$.
        Consider therefore $E$ and $E'$ as described in this lemma.
        
        First, we note that for the reads eliminated as part of $tr$, no $\rf$ edge associated with them will result in a consistent execution $E \sim E'$.
        \begin{itemize}
            \item If a crucial set $cr \in Cr(E, SC)$ consists only of the reads eliminated, then by $pwc(SC)$ we can obtain a consistent execution $E_{t} \sim E'$.
            \item By Lemma~\ref{lem:comp-gen-cond}, we do not need to address such $tr$. 
            \item We require every execution $E \sim E'$ such that the $\rf$ relations with the eliminated reads are not sufficient to make it consistent under $SC$.
            \item Thus, WLOG, we only consider the crucial set of $E$ to be of those reads which are not eliminated.
        \end{itemize}

        Now, by Corollary~\ref{cor:sra-incons-exec}, we must have at least one of the following true in $E$  
        \begin{tasks}(2)
            \task $\mo$ non-strict. 
            \task $\mo;\po$ cycle.
            \task $\rfi;[r_{x}];\po;[w_{x}]$ cycle.
            \task $\rb;[w_{x}];\po;[r_{x}]$ cycle.
            \task $\rb;\rfe;[r'_{x}];\po;[r_{x}]$ cycle.
            \task $[w];\mo;[w_{x}];\rfe;[r_{x}];\po$ cycle.
            \task $\rb;\hb;\rfe;[r_{y}];\po;[r_{x}]$ cycle.
            \task $\rb;\mo$ cycle.
            \task $\mo;\rf$ cycle.
        \end{tasks}

        We address each case separately. 
        \begin{description}
            \item[Case (a):] $\mo$ cycle in $E$.
            
            Since $st^{-} \cap w(P) = \phi$ we have $\mo(E) = \mo(E')$.
            This implies $\mo$ must also be non-strict in $E'$, making it inconsistent under $SRA$. 

            \item[Case (b):] $\mo;\po$ cycle in $E$ 
            
            By Lemma~\ref{lem:sc-crucial-exist}, a crucial set does not exist for $E$, but exists for $E'$, which implies $\neg \psf{SC}{tr}{P}$. 
        
            \item[Cases (c), (d):] $\rfi;[r_{x}];\po;[w_{x}]$ cycle and $\rb;[w_{x}];\po;[r_{x}]$ cycle in $E$.
            
            We can infer de-ordering of $r_{x}, w_{x}$ occurs, which by Lemma~\ref{lem:transf:rw-wr-eq-sep:sc} implies $\neg \psf{SC}{tr}{P}$.
        
            \item[Case (e):] $\rb;\rfe;[r'_{x}];\po;[r_{x}]$ cycle in $E$
        
            Consider the cycle in $E$ expanded as $[r_{x}];\rf^{-1};[w_{x}];\mo_{loc};[w'_{x}];\rfe;[r'_{x}];\po$. 
            This implies $E$ must have either $r_{x}$ or $r'_{x}$ in every crucial set under $SC$.
            The argument now follows using $E'$:
            \begin{itemize}
                \item First, note that the only corner case we need to address is when every crucial set $cr'$ for $E'$ has at least one of $r'_{x}$ or $r_{x}$.
                \item This is possible only when we have both $[r_{x}];\rb;[w1_{x}];\mo_{ext};[w_{y}];\po$ and $[r'_{x}];\rb;[w2_{x}];\mo_{ext};[w_{z}];\po$ cycles in $E'$.
                \item Since $\rb;\hb$ is irreflexive in $E'$ and since both $r_{x}, r'_{x}$ are of same memory. by Corollary~\ref{corr:mo-ext-flip}, we can obtain our $E'_{t}$ for which $r_{x}, r'_{x}$ belong to neither of these cycles.
                \item From Corollary~\ref{corr:mo-ext-flip}, we also know $\mo_{x}(E') = \mo_{x}(E'_{t}) = \mo_{x}(E_{t})$.
                \item Thus, $\rb;\rfe;[r'_{x}];\po;[r_{x}]$ cycle in $E_{t}$, requiring either $r_{x}$ or $r'_{x}$ to be in any of its crucial set.
                %\item To show this is true, we note that getting such $E'_{t}$ possibly change additional $\mo$ relations of the form $[w1_{x}];\mo_{!loc};[w]$, $[w2_{x}];\mo_{!loc};[w']$, $[w];\mo;[w_{y}]$ or $[w'];\mo;[w_{z}]$. This ensures we do not get another cycle with $r_{x}$ or $r'_{x}$ of the form $\rb;\mo_{ext};\po$. 
                %\item Second, we need to make sure that $[w_{x}];\mo_{loc};[w'_{x}]$ remains the same. This is true, as by Lemma~\ref{lem:sc-wr:p1}, the $\mo_{loc}$ relations with the writes $w_{x}$ and $w'_{x}$ do not change.
                %\item This gives us resultant $E'_{t}$, for which $r'_{x}$ or $r_{x}$ need not be in every crucial set $cr'$. 
                %\item However, the corresponding $E_{t}$ still has $[r_{x}];\rb;\rfe;[r'_{x}];\po$, for which at least one of $r_x$ or $r'_x$ must be part of its crucial set. 
                \item Since $cra(\rb;\mo_{ext};\rfe;\po) \nsubseteq cra(\mo;\rfe;\po)$, we can construct a minimal crucial set $cr'$ for $E'_{t}$ without $r_{x}, r'_{x}$.
                \item By Lemma~\ref{lem:crucial-based-unsafety}, we can infer $\neg \psf{SC}{tr}{P}$.
            \end{itemize}

            \item[Case (f)]: $[w];\mo;[w_{x}];\rfe;[r_{x}];\po$ cycle in $E$.
        
            We note that $r_{x}$ must be in every crucial set of $E$ under $SC$.
            Using $E'$, we proceed as follows: 
            \begin{itemize}
                \item We first note the only case we need to consider is when every crucial set $cr'$ of $E'$ must have $r_{x}$.
                \item This is possible only if $[r_{x}];\rb;\mo_{ext};\po$ forms a cycle in $E'$ (for the other composition, we can choose the other read to construct $cr'$).
                \item By Lemma~\ref{lem:sc-wr:p1}, we can obtain another execution $E'_{t}$ such that $r_{x}$ is not part of any cycle of the form $[r_{x}];\rb;\mo_{ext};\po$. 
                \item By Lemma~\ref{lem:sc-wr:p1} $E'_{t}$ will also have $\mo;[w_{x}](E') = \mo;[w_{x}](E'_{t})$.
                \item Thus, the corresponding $E_{t}$ will still have $[w];\mo;[w_{x}];\rfe;[r_{x}];\po$ cycle, implying $r_x$ must be in every crucial set of $E_{t}$. 
                \item For $E'_{t}$, since $cra(\rb;\mo_{ext};\rfe;\po) \nsubseteq cra(\mo;\rfe;\po)$, we can construct a minimal crucial set $cr'$ for $E'$ without $r_{x}$.
                \item Thus, by Lemma~\ref{lem:crucial-based-unsafety}, we have $\neg \psf{SC}{tr}{P}$.              
            \end{itemize} 
        
            \item[Case (g)]: $\rb;\hb;\rfe;[r_{y}];\po;[r_{x}]$ cycle in $E$
            
            Consider the cycle to be expanded as 
            \begin{align*} 
                [r_{x}];\rf^{-1};[w'_{x}];\mo_{x};[w_{x}];\hb;[w_{y}];\rfe;[r_{x}];\po.
            \end{align*}
            Here, at least one of $r_{x}$ and $r_{y}$ must be in the crucial set of $E$.
            Assuming $\mo;\hb$ irreflexive in $E$, we divide into sub-cases based on the $\mo_{ext}$ relation in $E'$ violating $a_{wr}$.
            Note that the only cycles we need to consider is when every crucial set $cr'$ of $E'$ has at least one of $r_{x}$ or $r_{y}$.
            We go case-wise on possible cycles such that either $r_{x}$ or $r_{y}$ or both exist as crucial reads, violating $a_{wr}$.
            \begin{description}
                \item[Subcase 0:] $[r_{x}];\rb;[w_{x}];\mo_{ext};[w_{y}];\po$ or $[r_{x}];\rb;[w_{x}];\mo_{ext};[w_{y}];\rfe;[r_{y}];\po$ cycle. 
                    \begin{itemize}
                        \item By Lemma~\ref{lem:sc-wr:p1}, we can obtain $E'_{t}$ such that $r_{x}$ is no longer a crucial read for the above composition. 
                        \item However, this implies we have $[w_{y}];\mo;[w_{x}]$ in $E'_{t}$ and corresponding $E_{t} \sim E'_{t}$.
                        \item This will result in $[w_{y}];\mo;[w_{x}];\hb$ cycle in $E_{t}$, which can be addressed via Case (b), or Case (f) \footnotemark.
                    \end{itemize}
                \item[Subcase 1:] $[r_{x}];\rb;[w'_{x}];\mo_{ext};[w'_{y}];\rfe;[r_{y}];\po$ cycle.
                    \begin{itemize}
                        \item Since $\rf^{-1}$ is functional, $w'_{y} = w_{y}$. 
                        \item From $E \sim E'$, we also have $\rb;[w_{x}];\mo;[w_{y}];\rfe;[r_{y}]\po$ cycle in $E'$.
                        \item If $[w_{x}];\hb;[w_{y}]$, then $E'$ is also inconsistent under $SRA$, violating our premise.
                        \item If $[w_{x}];\mo_{ext};[w_{y}]$, it reduces to Subcase 0 above.
                    \end{itemize}
                \item[Subcase 2:] $[r_{y}];\rb;[w'_{y}];\mo;[w'_{x}];\rfe;[r_{x}]\po$ cycle.
                    \begin{itemize}
                        \item From given, we have $[r_{x}];\rf^{-1};[w'_{x}];\mo_{loc};[w_{x}];\mo;[w_{y}];\rfe;[r_{y}];\po$ cycle in $E$.
                        \item From $E \sim E'$, we get $[w_{y}];\mo_{loc};[w'_{y}];\mo;[w'_{x}];\mo$ cycle for both $E$ and $E'$. 
                        \item This would make $E'$ inconsistent under $SRA$, violating our premise.
                    \end{itemize}
                \item[Subcase 3:] $[r_{x}];\rb;[w'_{x}];\mo_{ext};[w];\po$ cycle.
                    \begin{itemize}
                        \item By Lemma~\ref{lem:sc-wr:p1}, we have $E'_{t}$ such that $[r_{x}];\rb;\mo^{?};\hb$ irreflexive.
                        \item From Corollary~\ref{corr:mo-ext-flip}, we have $\mo_{x}(E') = \mo_{x}(E'_{t})$.
                        \item Thus, we still have $rb;\hb;\rfe;[r_{y}];\po;[r_{x}]$ cycle in $E_{t}$.
                        \item Since $cra(\rb;\mo;\rfe;\po) \nsubseteq cra(\rb;\mo;\po)$, $r_{x}$ need not be in every crucial set of $E'_{t}$. 
                        \item However, $r_{y}$ could still be needed for $cr'$, thus being a crucial set for both $E'_{t}$ and $E_{t}$.
                        \item We address this in the following sub-case.
                    \end{itemize}    
                \item[Subcase 4:] $[r_{y}];\rb;[w'_{y}];\mo_{ext};[w];\po$ cycle (this is the most tricky case).
                    \begin{itemize}
                        \item We first assume from Subcase 3 that $[r_{x}];\rb;\mo_{ext}^{?};\hb$ irreflexive in $E'$.
                        \item Now, from Lemma~\ref{lem:sc-wr:p1}, we can obtain $E'_{t}$ for which $[r_{y}];\rb;\mo^{?};\hb^{?}$ irreflexive.
                        \item From (g), we know that as long as $[w'_{x}];\mo_{x};[w_{x}]$ remains unchanged, $r_{x}$ or $r_{y}$ must be in every crucial set for $E$.
                        \item From $E$, we can infer $[w'_{x}];\mo;[w'_{y}]$ and $[w_{x}];\mo;[w'_{y}]$ in $E'$. %(may need constraint over write eliminations).
                        \item Since $E'$ is consistent under $SRA$, we can also infer $w \neq w_{x}$ and $w \neq w'_{x}$.
                        \item To get $E'_{t}$, we only possibly change other $\mo$ relations with a write $e$ such that $[w'_{y}];\mo;[e];\mo;[w];\mo$ forms a cycle.
                        \item From above, we can infer $e \neq w'_{x}$ and $e \neq w_{x}$, ensuring we would still have $[w'_{x}];\mo;[w_{x}]$ in $E'_{t}$ and in $E_{t} \sim E'_{t}$.
                        \item This ensures $\rb;\hb;\rfe;[r_{y}];\po;[r_{x}]$ cycle in $E_{t}$, implying $r_{x}$ or $r_{y}$ must still be in every crucial set of $E_{t}$.
                        \item We now show that we can have $E'_{t}$ such that $[r_{x}];\rb;\mo^{?};\hb^{?}$ irreflexive.
                        \item First, we know both $e$ and $w$ cannot be $w'_{x}$, ensuring $[w'_{x}];\mo$ or $\mo;[w'_{x}]$ is preserved over $E'_{t}$.
                        \item This ensures we have $[r_{x}];\rb;\hb$ irreflexive.
                        \item Now it can be possible for $[w'_{x}];\mo_{x};[w]$ and $[e];\po;[r_{x}]$ in $E'$, such that on changing $[e];\mo;[w]$ we get $[r_{x}];\rb;[w];\mo;[e];\po$ cycle in $E'_{t}$.
                        \item We show that we do not need to change $[e];\mo;[w]$, and instead we will have $[w'_{y}];\mo_{ext};[e]$ to change. 
                        \item First, $[e];\hb;[w'_{y}]$ is not possible, as $\mo;\hb$ irreflexive in $E'$.
                        \item Second, if $[w'_{y}];\hb;[e]$, then we have $[r_{x}];\rf^{-1};[w'_{x}];\mo;[w_{x}];\mo;[w'_{y}];\hb$ cycle in $E'$.
                        \item This violates our premise over $E'$ using Subcase 3.
                        \item Hence, we must have $[w'_{y}];\mo_{ext};[e]$, thereby allowing us to remove it to get $E'_{t}$ instead.
                        \item Thus, $[r_{x}];\rb;\mo_{ext}^{?};\hb$ is irreflexive in $E'_{t}$.
                        \item Since $cra(\rb;\mo;\rfe;\po) \nsubseteq cra(\rb;\mo;\po)$, both $r_{x}$ and $r_{y}$ need not be in every crucial set of $E'_{t}$.
                        \item Finally, from Lemma~\ref{lem:crucial-based-unsafety}, we have $\neg \psf{SC}{tr}{P}$.
                    \end{itemize}
            \end{description}

        \item[Cases (h), (i):] $\rb;\mo$ cycle and $\mo;\rf$ cycle in $E$.
        
            They must also form a cycle in $E'$, making it instead inconsistent under $SRA$.
        \end{description}
        
        \footnotetext{
            Though we may have other cycles in $E'_{t}$ that make it inconsistent under $TSO$, Cases (b), (e) do not leverage this. Hence it can be deferred to these cases.
        }

        Hence proved.
    \end{proof}

    %%NOT RELEVANT FOR TASE, BUT FOR PHD THESIS -----------------------------------------------------------------------------------------------------------------------------
   
    %The following can be inferred from the model of $TSO$ and Lemma~\ref{lem:sc-crucial-exist}.
    %\begin{corollary}
    %    \label{cor:tso-crucial-exist}
    %    Given an inconsistent execution $E$ under $TSO$ such that $\mo$ is a strict total order such that $\mo;\po$ is irreflexive, a crucial set exists.
    %\end{corollary}
%
    %\begin{proof}
    %    Inferred directly from Prop~\ref{prop:sc-cra}.
    %\end{proof}

    %
    %\begin{restatable}{theorem}{thmsratsocomp}
    %\label{thm:sra-tso-comp}
    %    For transformation-effects $P \mapsto_{tr}$ with no write-modification effects ($we/wi$) and $iwri$, we have $\comp{SRA}{TSO}$.
    %\end{restatable}
%
    %\pending{The proof does not require constraint over $iwri$.}
    %\begin{proof}
    %    All cases are addressed by Theorem~\ref{thm:sra-sc-comp}.
    %    Case (a), (b), (c), (d), (e), (f) from Corollary~\ref{cor:sra-incons-exec} are addressed the same way. 
    %    Case (g) would still require excluding $iwri$ as by Prop~\ref{prop:cons-sra-incons-tso} we can have $\rb;\mo;[u];\po$ cycle in $E'$.
    %    Hence, proved.
    %\end{proof}

%% file: 5.Appendix/tso_appendix.tex
\section{Proving TSO Complete w.r.t. SC}

    \label{sec:tso-proofs}
    The following is inferred from Lemma~\ref{lem:sc-incons-exec}.
    \begin{restatable}{corollary}{cortsoincons}
            \label{cor:tso-incons-exec}
            For any execution $E$ inconsistent under $TSO$, at least one of the following is true.
            \begin{tasks}(2)
                \task $\mo$ cycle.
                \task $\mo;\po$ cycle.
                \task $\rfi;\po$ cycle.
                \task $\rb;\po$ cycle.
                \task $\rb;\rfe;\po$ cycle.
                \task $\mo;\rfe;\po$ cycle.
                \task $\rb;\hb;\rfe;\po$ cycle.
                \task $\rb;\mo$ cycle.
                \task $\mo;\rf$ cycle.
                \task $\rb;\mo_{ext};\rfe;\po$ cycle.
                \task $\rb;\mo;[u];\po$ cycle.
            \end{tasks}
    \end{restatable}

    \begin{proof}
        Checking the $TSO$ rules from Def~\ref{def:tso-model}, we know that not all compositions of the form $\rb;\mo;\hb$ need to form a cycle.
        Instead, we only require $\rb;\mo_{ext};\rfe;\po$ and $\rb;\mo_{ext};[u];\po$ would need to be irreflexive instead.
        Hence proved.  
    \end{proof}

    \thmtsosccomp*

    \begin{proof}
        
        From Theorem~\ref{thm:sra-sc-comp} and Corollary~\ref{cor:tso-incons-exec}, all but cases (j) and (k) remain.
        Since $M = TSO$, from Prop~\ref{prop:tso-cons-sc-incons}, the only possible cycle in $E'$ can be of the form $\rb;\mo_{ext};\po$ 
       
        \begin{description}
            \item[Case (j)] $[r_{x}];\rb;[w_{x}];\mo_{ext};[w_{y}];\rfe;[r_{y}];\po$ cycle in $E$.
        
            We first show that for such a case reduces to program context specified in Def~\ref{def:tso-effect-const}.
            \begin{itemize}
                \item If $[w_{x}];\po;[r_{x}]$ in $E$, then it reduces to Case (d) as $\rb;\po$ cycle in $E$. 
                \item If $[r_{x}];\po;[w_{x}]$ in $E$, then it reduces to Case (f) as $\mo;\rfe;\po$ cycle in $E$.
                \item Thus, we must have no $\po$ between $w_{x}$ and $r_{x}$ in $P$.
                \item From the cycle in $E$ and same reasoning as above, we can also infer $w_{x}$ has no $\po$ with $r_{y}$.
                \item The rest lines up with the first program context specified in Def~\ref{def:tso-effect-const}.
                \item Thus, such a $tr$ must not involve $\tuwri$ such that $[w_{y}];\po;[r_{x}]$ in $E'$.
            \end{itemize}

            We note that the only case we need to address is when every crucial set $cr'$ of $E'$ has at least one of $r_{x}$ or $r_{y}$.
            This would be possible if at least one of the following cycles exist in $E'$
            \begin{tasks}(2)
                \task $[r_{y}];\rb;[w'_{y}];\mo_{ext};[w3];\po$.
                \task $[r_{x}];\rb;[w'_{x}];\mo_{ext};[w2];\po$.
            \end{tasks}

            We now proceed case-wise for each.
            \begin{description}
                \item[Subcase 1:] $[r_{y}];\rb;[w'_{y}];\mo_{ext};[w3];\po$ cycle in $E'$.
                
                We first show that eliminating the above cycle preserves the cycle of Case (j).
                In doing so, we also show which cycles involving $r_{x}$ cannot exist. 
                This will be useful when addressing the next sub-case.
                %\pending{Change this case slightly, as we cannot use $\tuwri$ here for our purpose.}
                %We first show that if sub-case (a) is true, then $r_{x}$ must also be in every crucial set of $E$.
                %\begin{itemize}
                %    \item From $E$, we have $[w_{x}];\mo;[w_{y}]$ and $[r_{y}];\po;[r_{x}]$.
                %    \item From $E'$, we have $[w_{y}];\mo_{loc};[w'_{y}];\mo;[w3]$ and $[w3];\po;[r_{y}]$.
                %    \item From $st^{-} = \phi$, $\mo$ strict total in $E'$ and $E' \sim E$, we must have $[w_{x}];\mo;[w'_{y}];\mo;[w3]$ in $E$. %(w-elimination usage).
                %    \pending{The exact constraints $\tuwri$ places to warrant $[w3];\po;[r_{x}]$, must be explained. It is not hard though.}
                %    \item Since $tr$ does not involve $\tuwri$ we also have $[w3];\po;[r_{y}]$ in $E$, giving us $[w3];\po;[r_{x}]$ in $E$. 
                %    \item Thus, we also have $[r_{x}];\rb;[w_{x}];\mo;[w3];\po$ cycle in $E$.
                %    \item From $cra(\rb;\mo;\po)$, $r_{x}$ must be in every crucial set of $E$.         
                %\end{itemize}

                \begin{itemize}
                    \item We can use Lemma~\ref{lem:sc-wr:p1} to obtain $E'_{t}$ such that $[r_{y}];\rb;\mo;\hb$ irreflexive. 
                    \item From Lemma~\ref{lem:sc-wr:p1}, we have $\mo;[w_{y}](E') = \mo;[w_{y}](E'_{t})$ (where $[w_{y}];\rf;[r_{y}]$).
                    \item Thus, we will still have $[r_{x}];\rb;\mo_{ext};[w_{y}];\rfe;[r_{y}];\po$ cycle in $E_{t}$.
                \end{itemize}
                Now we show certain cycles involving $r_{x}$ are irreflexive in $E'_{t}$.
                We first note that the only cause for an additional cycle with $r_{x}$ will be when we have to address cycles of the form $[w'_{y}];\mo;[e];\mo;[w3];\mo$ where $mem(r_{x}) = mem(e)$. or $mem(r_{x}) = mem(w3)$. 
                \begin{itemize}
                    \item If $e=w'_{x}$ such that $[w'_{x}];\rf;[r_{x}]$, then we have $[w'_{x}];\mo;[w_{y}];\mo;[w'_{y}];\mo$ cycle in $E$.
                    \item Since $tr$ has no $\we$, this will make $E'$ inconsistent under $TSO$, violating our premise.
                    \item Similar is the case when $w3=w'_{x}$.
                    \item Thus, we know that $\mo;[w'_{x}](E') = \mo;[w'_{x}](E'_{t})$.
                    \item Since $[r_{x}];\rb;\hb$ irreflexive in $E'$, it must be irreflexive for $E'_{t}$.
                    \item If $r_{x} \in rmw(p(E))$, then $e= r_{x}$ or $w3=r_{x}$ would give $[w'_{y}];\mo;[r_{x}]$ in $E'$. 
                    \item This would imply $[r_{x}];\rb;\mo$ cycle in $E$ as well as $E'$ (Case (h)).
                    \item Thus, we can infer $[r_{x}];\rb;\mo$ irreflexive in $E'_{t}$.
                \end{itemize}
                All that remains is the possibility of $[r_{x}];\rb;\mo;\rfe;\po$ cycle in $E'_{t}$.
                We show in the next Subcase that this is okay. 

                \item[Subcase 2:] $[r_{x}];\rb;[w'_{x}];\mo_{ext};[w2];\po$ cycle in $E'$. 
                
                For this case, we show that using Lemma~\ref{lem:sc-wr:p1} to obtain $E'_{t}$ $[r_{x}];\rb;\mo^{?};\hb$ irreflexive preserves $[r_{x}];\rb;[w_{x}];\mo;[w_{y}];\rfe;[r_{y}];\po$ cycle in $E_{t}$.
                \begin{itemize}
                    \item First, since $tr$ does not involve $\tuwri$, we can infer $tid(w2) \neq tid(w_{y})$, implying $w_{y} \neq w2$.
                    \item Next, let us assume we have $w'_{x} = w_{x}$, and that $[w_{y}];\mo;[w2]$, $[w_{x}];\mo;[w_{y}]$ in $E'$.
                    \item This would, by Lemma~\ref{lem:sc-wr:p1} require us changing one of them to give us $E'_{t}$ (and $E_{t}$). 
                    \item Now we show $[w_{y}];\mo_{ext};[w2]$ in $E'$.
                    \item IF $[w2];\hb;[w_{y}]$ then we will have $\mo;\hb$ cycle in $E'$.
                    \item If $[w_{y}];\hb;[w2]$, then from $E \sim E'$, we have $[r_{x}];\rb;[w_{x}];\mo;[w_{y}];\hb;[w2];\po$ cycle in $E'$.
                    \item Since $tid(w2) \neq tid(w_{y})$, we must have $\rb;[w_{x}];\mo;[w_{y}];\hb^{?};\rfe;\po$ cycle in $E'$.
                    \item This is not possible, as $E'$ is inconsistent under $TSO$.
                    \item Even from Subcase 1, since $\mo;[w_{y}]$ remains unchanged across $E'_{t}$, the same cycle exists in $E'$.
                    \item Thus, we must have $[w_{y}];\mo_{ext};[w2]$, which can be changed instead to obtain $E'_{t}$.
                    \item This guarantees we can have $[r_{x}];\rb;[w_{x}];\mo;[w_{y}];\rfe;\po$ cycle in $E_{t}$.
                \end{itemize}
            \end{description}
        
            We are now left with $E_{t}$ which can have either $r_{x}$ or $r_{y}$ in its crucial set.
            Finally, we show that there exists a crucial set $cr'$ of $E'_{t}$ which is not for $E_{t}$.
            \begin{itemize}
                \item First, we know a $cr'$ exists such that $r_{x} \notin cr'$.
                \item If $r_{y} \notin cr'$, we are done. 
                \item If $r_{y} \in cr'$ in every such $cr'$, then it is only due to $[r_{y}];\rb;\mo^{?};\hb^{?}$ cycle in $E'_{t}$.
                \begin{itemize}
                    \item If $[r_{y}];\rb;\mo;\po$ is already a cycle in $E'_{t}$, then we address Subcase 1 first, ensuring $[r_{y}];\rb;\mo^{?};\hb$ irreflexive.
                    \item The corresponding $E'_{t}$ from Subcase 1 will still have $[r_{x}];\rb;\hb$ and $[r_{x}];\rb;\mo$ irreflexive.
                    \item The case of $[r_{x}];\rb;\mo;\rfe;\po$ cycle can be addressed the same way as Subcase 2.
                    \item Next, since $w_{y} \neq w2$, we must have $[r_{y}];\rb;\hb$ irreflexive in $E'_{t}$.
                    \item If $[r_{y}];\rb;[w2];\mo;\hb$ cycle in $E'_{t}$, then $[w_{x}];\mo;[w_{y}];\mo;[w2];\mo$ cycle in $E'_{t}$, which cannot be the case. 
                    \item If $[r_{y}];\rb;[w2];\mo$ cycle in $E'_{t}$ (if $r_{y}$ is read-modify-write), then we have $[r_{x}];\rb;\mo;[r_{y}];\po$ cycle in $E$ (Case (k)) or $\mo;[w_{y}];\rfe$ cycle in $E$ Case (i).
                \end{itemize}
                \item To conclude, if $r_{y}$ is in every $cr'$, then it must reduce to Case (i), (k) .
                \item Otherwise, by Lemma~\ref{lem:sc-wr:p1}, we have $\neg \psf{SC}{tr}{P}$.
            \end{itemize}

        \item[Case (k):] $[r_{x}];\rb;\mo_{ext};[u];\po$ cycle in $E$.
        \begin{itemize}
            \item First, note that $\rb;\mo;[u];\po$ cycle cannot exist in $E'$.
            \item Thus, at least $(e, r_{x}) \in \po^{-}$.
            \item Now, if $r_{x}$ need not be in every crucial set of $E'$, we are done. 
            \item Otherwise, we have $[r_{x}];\rb;[w_{x}];\mo_{ext};[w];\po$ cycles in $E'$.
            \item If $[w];\po;[r_{x}]$ in $P$, then $tr$ involves $\tuwri$. %Ridding the need for write elimination 
            \item Otherwise, we can divide $tr$ into $P \mapsto P_{tmp} \mapsto P'$, where $[w];\po;[r_{x}]$ and satisfies one of the program contexts in Def~\ref{def:tso-effect-const} for $\tuwri$.
            \item Hence, Case (k), reduces to one which involves $\tuwri$.
        \end{itemize}

        %We show that we can use Lemma~\ref{lem:sc-wr:p1} to obtain $E'_{t}$ and $E_{t}$ such that $r_{x}$ must still be in every crucial set of $E_{t}$.
        %\begin{itemize}
        %    \item By Lemma~\ref{lem:sc-wr:p1}, to obtain our $E'_{t}$ (and $E_{t}$), we know the only changes are made to relations of the form $[w_{x}];\mo;[e]$ and $[e];\mo;[w]$ (apart from $[w_{x}];\mo;[e]$).
        %    \item If $e = u$, then either $[w_{x}];\mo;[u]$ or $[u];\mo;[w]$ would need to change due to $[w];\mo;[w_{x}];\mo;[u];\mo$ cycle.
        %    \item If $[u];\hb;[w]$, then we have either $\rb;\mo;[u];\po$ or $\rb;\mo;\rfe;\po$ cycle in $E'$, making it inconsistent under $TSO$.
        %    \item If $[w];\hb;[u]$, then $\mo;\hb$ cycle in $E'$, violating our premise over $E'$.
        %    \item Thus, the remaining case is when $[u];\mo_{ext};[w]$, which we can remove to obtain our $E'_{t}$.
        %    \item This preserves $\mo;[u]$ across $E'_{t}$ and thus, also across $E_{t}$.
        %    \item By Corollary~\ref{corr:mo-ext-flip}, $\mo_{x}$ is also preserved \footnotemark.
        %    \item Hence, we retain $[r_{x}];\rb;\mo;[u];\po$ cycle in $E_{t}$, requiring $r_{x}$ to be in its every crucial set.
        %    \item By construction, the same is not required for $E'_{t}$.
        %    \item Therefore, by Lemma~\ref{lem:crucial-based-unsafety} we have $\neg \psf{SC}{tr}{P}$. 
        %\end{itemize} 
%
        %\footnotetext{We can claim Corollary~\ref{corr:mo-ext-flip} holds for our cases as the $E'$ we address will always be consistent under $TSO$.
        %The part where Case (j) defers to Case (k) is also when we have cycles in $E$ as opposed to $E_{t}$ or $E'_{t}$.}
        \end{description}

        Hence, proven.
      
    \end{proof}

    %------------------------------------------------------------------------------------------------------------------------------